%% file: satur.tex
\documentclass{article}
\usepackage{etex}
\usepackage[a4paper]{geometry}
\usepackage{makeidx}
\makeindex

\input{parent-preambule.tex}

\input{parent-en-preambule.tex}

\input{preambule.tex}

\date{} \title{A Quasi-Linear Time Algorithm Deciding Whether Weak
  Büchi Automata Reading Vectors of Reals Recognize Saturated
  Languages} \author{Arthur Milchior}
\begin{document}

\maketitle
\begin{abstract}
  This work considers weak deterministic Büchi automata reading
  encodings of non-negative $d$-vectors of reals in a fixed base. A
  saturated language is a language which contains all encoding of
  elements belonging to a set of $d$-vectors of reals.  A Real Vector
  Automaton is an automaton which recognizes a saturated language. It
  is explained how to decide in quasi-linear time whether a minimal
  weak deterministic Büchi automaton is a Real Vector Automaton. The
  problem is solved both for the two standard encodings of vectors of
  numbers: the sequential encoding and the parallel encoding. This
  algorithm runs in linear time for minimal weak Büchi automata
  accepting set of reals. Finally, the same problem is also solved for
  parallel encoding of automata reading vectors of relative reals.
\end{abstract}
\section{Introduction}

This paper deals with logically defined sets of vector of numbers
encoded by Büchi deterministic automata.  The sets of vectors of
integers whose encodings in base $b$ are recognized by a finite
automaton are called the $b$-recognizable sets.  By \cite{bruyere},
the $b$-recognizable sets are exactly the sets which are
$\fo{\Z;+,<,V_{b}}$-definable, where $V_{b}(n)$ is the greatest power
of $b$ dividing $n$. It was proven in \cite{semenov-theorem,Cobham}
that the $\fo{\N;+}$-definable sets are exactly the sets which are
$b$- and $b'$-recognizable for every $b\ge 2$.

Those results have then been extended to results about sets of vectors
of reals recognized by a Büchi automata.  The notion of Büchi automata
is a formalism which describes languages of infinite words, also
called $\omega$-words. The Büchi automata are similar to the finite
automata. The main difference between the two kinds of automata is
that a finite automaton accepts a finite word if it admits a run
ending on accepting states, while a Büchi automaton accepts an
infinite word it it admits a run in which an accepting state appears
infinitely often.

One of the main differences between finite and Büchi automata is that
finite automata can be determinized while deterministic Büchi automata
are less expressive than Büchi automata. For example, the language
$L_{\text{fin }a}$ of words containing a finite number of times the
letter $a$ is recognized by a Büchi automaton, but is not recognized
by any deterministic Büchi automaton. This statement implies, for
example, that no deterministic Büchi automaton recognizes the set of
reals of the form $nb^{p}$ with $n\in\N$ and $p\in\Z$, that is, the
reals which admits no encoding in base $b$ with a finite number of
non-0 digits.

Another main difference between the two classes of automata is that
the class of languages recognized by finite automata is closed under
complement while the class of languages recognized by deterministic
Büchi automata is not closed under complement. For example,
$L_{\text{inf }a}$, the complement of $L_{\text{fin }a}$, is recognized
by a deterministic Büchi automaton.

\paragraph{}
Real numbers are naturally encoded, in a base $b>1$, as a sequence of
digits in $\set{0,\dots,b-1}$ and a separator symbol $\realDot$.  That
is, as a word over the alphabet
$\set{0,\dots,b-1,\realDot}$. Similarly, a $d$-vector of real numbers
can be encoded as a word over alphabet
$\set{0,\dots,b-1}^{d}\cup\set\realDot$, where $d$ digits are read
simultaneously, one for each element of the vector.  This is call the
$d$-parallel encoding of the vector. A $d$-vector can also be encoded
as a word over alphabet $\set{0,\dots,b-1,\realDot}$, assuming that
the digits in position $i$ modulo $d$ corresponds to the digits of the
$i$-th element of the vector. This is call the sequential encoding of
the vector of digits. The cardinality of the alphabet of parallel
encoding is exponentially bigger than the cardinality of the alphabet
of sequential encodings, thus, sequential encodings may be preferred
for practical purposes. Parallel encoding leads to simpler notation,
hence, most of the litterature consider parallel encodings. We
consider both encodings in this paper.

A language $L$ is said to be \emph{saturated} if, given a vector
$\mathbf r\in\mathbb R^{d}$, the set of its encoding in base $b$ is
either included in $L$ or disjoint from $L$. A Real Vector Automaton
(RVA, See e.g. \cite{weak-R-+-vector}) is an automaton of alphabet
$\set{0,\dots,b-1}^{d}\cup\set\realDot$ which recognizes a saturated
language. Here $d$ is the dimension of the vector that the automata
read. In the case where the dimension $d$ is 1, those automata are
called Real Number Automata (RNA, See e.g. \cite{real}).

The sets of vectors of reals whose encodings in base $b$ is recognized
by a RVA are called the $b$-recognizable sets. By \cite{LinArCon},
they are exactly the $\fo{\R,\Z;+,<,X_{b},1}$-definable sets. The
logic $\fo{\R,\Z;+,<,X_{b},1}$ is the first-order logic over reals
with a unary predicate which holds over integers, addition, order, the
constant one, and the function $X_{b}(x,u,k)$. The function
$X_{b}(x,u,k)$ holds if and only if $u$ is equal to some $b^{n}$ with
$n\in\Z$ and there exists a encoding in base $b$ of $x$ whose digit in
position $n$ is $k$. That is, $u$ and $x$ are of the form:
\begin{eqnarray*}
  \begin{array}{lllllllllll c  lllllllllllll}
    u=&0&\dots&0&\realDot&0&\dots&0&1&0&\dots
    &\text{ or }&
    u=&0&\dots&0&1&0&\dots&0&\realDot&0&\dots \\
    x=& &\dots& &\realDot& &\dots& &k& &\dots
    &&
    x=& &\dots& &k& &\dots& &\realDot& &\dots& 
  \end{array}
\end{eqnarray*}
A weak deterministic Büchi automaton is a deterministic Büchi
automaton whose set of accepting states is a union of strongly
connected components. A set is said to be weakly $b$-recognizable if
it is recognized by a weak automaton in base $b$.  By
\cite{weak-R-+-vector}, a set is $\fo{\R,\Z;+,<}$-definable if and
only if its set of encodings is weakly $b$-recognizable for all
$b\ge2$. The weak deterministic Büchi automata are less expressive
than the deterministic Büchi automata. For example, the language
$L_{\text{inf }a}$ of words containing an infinite number of $a$ is
recognized by a deterministic Büchi automaton but is not recognized by
any weak deterministic Büchi automaton. This implies that, for
example, no weak deterministic Büchi automaton recognizes the set of
reals which are not of the form $nb^{p}$ with $n\in\N$ and $p\in\Z$,
since those reals are the ones whose encodings in base $b$ contains an
infinite number of non-$0$ digits.  Furthermore, by
\cite{minimal-buchi}, weak deterministic Büchi automata can be
efficiently minimized.

\paragraph{}
In this paper, we show that we can efficiently decide whether a weak
Büchi automaton accept a saturated set of vectors of
integers. Furthermore, we give an algorithm for automata reading
parallel encoding and for automata reading sequential encoding.

We recall standard definition in \autoref{sec:def}. We introduce
encoding of sets of vectors of numbers in \autoref{sec:encoding}. We
introduce Büchi automata in \autoref{sec:automata}. We formalize how
we compute the complexity of an algorithm in
\autoref{sec:time-space}. We study automata reading vectors of numbers
in \autoref{sec:d-seq-par}. We study how to transform words and
automata in \autoref{sec:fix-01}. We characterize the parallel RVA in
\autoref{sec:char-RVA-par} and the sequential RVA in
\autoref{sec:char-RVA-seq}. We explain how to decide whether an
automaton is a RVA in \autoref{sec:algo-RVA}. The case of sets
containing negative reals is discussed in \autoref{sec:neg}.

\section{Standard definitions}\label{sec:def}
We now give the standard definitions used in this paper.
\paragraph{Numbers.}
Let $\N$\index{N@$\N$} and $\R$ \index{R@$\R$} denote the set of
non-negative integers and the set of reals, respectively. For
$R\subseteq \R$, let $R^{\ge 0}$ denote the set of non-negative
elements of $R$.  For $n\in\N$, \index{n@$\left[n\right]$ for
  $n\in\N$.} let $[n]$ represent $\set{0,\dots,n}$. For $m\in\N^{>0}$,
let $(n\mod m)$ represents the only integer $i\in[m-1]$ such that
$n\equiv i \mod m$.
\paragraph{Sets.}
For $S$ and $T$ two sets, let
$S\otimes T=\set{(s,t)\mid s\in S,t\in T}$ be the set of ordered pair
containing an element of $S$ and an element of $T$. Let $\card{S}$ be
the cardinality of $S$.  For $d\in\mathbb N$, let $S^{d}$ be the set
of $d$-vectors of elements of $S$ for $d\in\mathbb N$. The $d$-vectors
are denoted $\tu s=(s_{0},\dots,s_{d-1})$ with each $s_{i}\in S$. The
$d$-vector $(0,\dots,0)$ is denoted \index{0@$\tu 0$}$\tu 0$. 
\paragraph{Finite and infinite words.}
An \emph{alphabet}\index{Alphabet} is a finite set, its elements are
called \emph{letters}\index{Letter}. A finite word\index{Finite word}
over the alphabet $\alphabet$ is a finite sequence of letters of
$\alphabet$.  An $\omega$-word\index{$\omega$-word} over the alphabet
$\alphabet$ is an infinite sequence of letters of $\alphabet$. The
empty word is denoted $\epsilon$\index{e@$\epsilon$}.  A set of finite
(respectively $\omega-$) words of alphabet $\alphabet$ is called a
language\index{Language}\index{Language@$\omega$-language}
(respectively, an $\omega$-language) over alphabet $\alphabet$.

Let $w$ be a word, its length is denoted $\length{w}$, it is either a
non-negative integer or the cardinality of $\N$.  For
$n\in[\length{w}-1]$, let $w[n]$ denote the $n$-th letter of $w$. For
$v$ a finite word, let $u=vw$ be the
\emph{concatenation}\index{Concatenation of words} of $v$ and of $w$, that is,
the word of length $\length{v}+\length{w}$ such that $u[i]=v[i]$ for
$i\in[\length{v}-1]$ and $u[\length{v}+i]=w[i]$ for
$i\in[\length{w}-1]$. Let $\prefix w{n}$\index{prefix@$\prefix w{n}$ -
  Prefix of $w$ of length $n$} denote the \emph{prefix} of $w$ of
length $n$, that is, the word $u$ of length $n$ such that $w[i]=u[i]$
for all $i\in[n-1]$. Similarly, let
$\suffix w{n}$\index{suffix@$\suffix w{n}$ - Suffix of $w$ without its
  $n$ first letters.} denote the \emph{suffix} of $w$ without its
$n$-th first letters, that is, the word $u$ such that $u[i]=w[i+n]$
for all $i\in[\length{w}-n]$.  Note that $w=\prefix wi\suffix wi$ for
all $i\in[\length{w}-1]$.
\paragraph{Languages}
A language is a set of words. Let $L$ be a language of finite words
and let $L'$ be either an $\omega$-languages or a language of finite
words.  Let $LL'$ be the set of concatenations of the words of $L$ and
of $L$. For $i\in\N$, let $L^{i}$ be the concatenations of $i$ words
of $L$.  Let $L^{*}=\bigcup_{i\in\N} L^{i}$\index{*@$L^{*}$}, more
generally, for $d,j\in\N$, let
$L^{d\N+j}=\bigcup_{i\in\N} L^{di+j}$\index{*@$L^{d\N+j}$} and
\index{+@$L^{+}$}$L^{+}=\bigcup_{i>0} L^{i}$. If $L$ is a set of
non-empty word, let $L^{\omega}$ be the set of infinite sequences of
elements of $L$. Finally, let $L^{\infty}=L^{*}\cup L^{\omega}$.
\section{Encoding of set of vectors of numbers}\label{sec:encoding}
In this section we explain how to encode sets of vectors of numbers
using languages. We consider natural and real numbers in
\autoref{susce:enc-reals}. We consider the special case of rationals
in \autoref{sec:rati-01}. We then consider vectors of reals in
\autoref{sec:enc-vector}. Finally, we consider sets of vectors of
reals in \autoref{subsec:enc-set-reals}.
\subsection{Encoding of numbers}\label{susce:enc-reals}
Let us now consider the encoding of numbers in an integer base
$b\ge2$. Let $\digitSet$ be equal to $[b-1]$, it is the set of
digits. The base $b>1$ is fixed for the remaining of this paper.
Formally, for $v\in\digitSet^{*}$ and $w\in\digitSet^{\omega}$:
\begin{eqnarray*}
  \wordToNatural{v}=\sum_{i=0}^{\length{v}-1}b^{\length{v}-1-i}v[i]\text{
  and }
  \wordToFractional{w}=\sum_{i=0}^{\infty}b^{-i-1}w[i].
\end{eqnarray*}
\index{.@$\wordToNatural{.}$}
\index{.@$\wordToReal{.}$}\index{.@$\wordToFractional{.}$}\index{Natural
  part of a word of
  $\digitSet^{*}\realDot\digitSet^{\omega}$}\index{Fractional part of
  a word of $\digitSet^{*}\realDot\digitSet^{\omega}$} Let $w$ be an
$\omega$-word with exactly one $\realDot$. It is of the form
$w=\natPart{w}\realDot\fraPart{w}$, with $\natPart{w}\in\digitSet^{*}$
and $\fraPart{w}\in\digitSet^{\omega}$. The word $\natPart{w}$ is
called the natural part of $w$ and the $\omega$-word $\fraPart{w}$ is
called its fractional part. We then define:
\begin{equation*}
  \wordToReal{\natPart{w}\realDot \fraPart{w}}=\wordToNatural{\natPart{w}}+\wordToFractional{\fraPart{w}}.
\end{equation*}
Examples of representation of numbers are now given.
  \begin{equation*}
    \arraycolsep=0.5pt
    \begin{array}{rclp{2mm}rclp{2mm}rclp{2mm}rclp{2mm}rcl}
      \wordToFractional[2]{(10)^{\omega}}&=&\frac 23 &
      &\wordToFractional[2]{(01)^{\omega}}&=&\frac13&
      &\wordToFractional[2]{0(10)^{\omega}}&=&\frac13&
      &\wordToFractional[2]{0(1)^{\omega}}&=&\frac{1}{2}&
      &\wordToFractional[2]{1(0)^{\omega}}&=&\frac12
      \\
      \wordToNatural[2]{10}&=&2&
      &\wordToNatural[2]{1}&=&1&
      &\wordToNatural[2]{01}&=&1&
      &\wordToNatural[2]{\epsilon}&=&0&
      &\wordToNatural[2]{00000}&=&0
      \\
      \multicolumn{5}{r}{\wordToReal[2]{10\realDot(10)^{\omega}}}&=&\frac{8}{3}&
      &\wordToReal[2]{\realDot0(1)^{\omega}}&=&\frac{1}{2}&
      &\multicolumn{5}{r}{\wordToReal[2]{00000\realDot 1(0)^{\omega}}}&=&\frac{1}{2}.
    \end{array}
  \end{equation*}
\paragraph{Pair-encoding}
A word $w\in\digitDotSet^{\infty}$ can equivalently be encoded as a
pair $\pair{w}{S}$ where $w\in\digitSet^{\infty}$ and $S\subseteq\N$. This
pair represents the word of length $\length{w}+\card S$, such that
$\pair{w}{S}[i]$ is $\realDot$ if $i\in S$, otherwise it is
$w\left[i-\card{\set{s\in S\mid s<i}}\right]$. Intuitively, for
$i\not\in S$, the $i$-th letter of $\pair{w}{S}[i]$ is the letter of $w$ at
a position $j$ such that $j+k=i$, where $k$ is the number of
$\realDot$'s before $i$. For example
$\pair{(10)^{\omega}}{\set2}=10\realDot(10)^{\omega}$,
$\pair{01}{\set0}=\realDot01$ and $\pair{(01)^{\omega}}{\emptyset}=(01)^{\omega}$.
The representation $\pair{w}{S}$ is called a
pair-encoding. 


If $\pair{w}{S}$ is an encoding of a real, or an encoding of a factor
of a real, $S$'s cardinality is at most 1. In particular, the pair
encoding of a word of $\digitSet^{*}\realDot\digitSet^{\omega}$ is of
the form $\pair{w}{\set s}$ with $w\in\digitSet^{\omega}$ and
$s\in\N$.  Note however that in order to check whether an automaton is
a RVA, it must be checked that it rejects every words whose number of
$\realDot$'s is not 1. Therefore, the cases where $S$ is not a
singleton must be considered.
\subsection{Encoding of rationals}\label{sec:rati-01}
We now recall a basic fact about encoding of rationals.
\begin{theorem}[\cite{number-theory-hardy}]\label{theo:rat}
  Let $q\ge 0$ a real.  Let $l=\ceil{\log_{b}(q+1)}$ and $l'\in\N$.

  The number of encoding of $q$ with a natural part of length $l'$ is:
  \begin{itemize}
  \item $0$ if $l'<l$,
  \item $2$ if $l'\ge l$ and if $q$ admits a decomposotion of the form
    $nb^{p}$ with $p\in\mathbb Z$ and $n\in\mathbb N$ such that
    $n\not\equiv 0\mod b$, and
  \item $1$ otherwise.
  \end{itemize}
  In the seconde case, the two encodings are of the form:
  \begin{equation}\label{eq:theo:rat}
    (va(b-1)^{\omega},\set i)\text{ and } (v(a+1)0^{\omega},\set i),
  \end{equation}
  with $v\in\digitSet^{*}$ and $a\in\digitSet\setminus\set{b-1}$.
\end{theorem}
This theorem illustrates that pair-encoding leads to shorter
statements. Indeed, with the standard-encoding, Equation
\eqref{eq:theo:rat} would require to consider three cases, depending
on whether $p<0$, $p=0$ or $p>0$. Note that the condition
$n\not\equiv 0\mod b$ ensures that $q\ne 0$.

\subsection{Encoding of vectors of reals.}\label{sec:enc-vector}
It is now explained how to encode $d$-vectors of real numbers.  In
this paper, we fix a positive integer constant $d$. In the remaining
of this paper, we only consider sets of dimension 1 or $d$.

There exists two standard encodings of vectors of numbers. The
parallel one and the sequential one. A parallel encoding consists in a
sequence of $d$-vector of digits. A sequential encoding consists in a
sequence of digits. This sequence contains alternatively a digit of
the zeroth number, a digit of the first number, up to a digit of the
$(d-1)$-th number.  In both cases, exactly one dot appear in the
sequence, to separate the natural part from the fractional part.

The alphabet of sequential encoding contains $(b+1)$ letters while the
alphabet of parallel encodings contains $(b^{d})+1$ letters. Thus,
sequential encoding allow to create smaller automata, as shown in
Example \ref{ex:exp-smaller}. However, parallel encoding leads to
notations which are more compact. Since parallel encoding are more
standards and lead to simpler proofs.
\paragraph{Parallel encodings}\label{sec:par-enc}
We now introduce the notion of parallel encoding of a $d$-vector of
numbers.  Let $\digitSetDim=\digitSet^{d}$, be the set of $d$-vectors
of digits. For $\tu w\in\digitSetDim^{\infty}$, and $0\le i<d$,
$w_{i}$ denote the unique word such that
$\length{w_{i}}=\length{\tu w}$ and such that for
$0\le k<\length{\tu w}$, $(w_{i})[k]=(\tu w[k])_{i}$.  Similarly, for
$\pair{\tu w}{S}\in\left(\digitDotSetDim\right)^{\infty}$,
$\pair{\tu w}{S}_{i}$ denote $\pair{w_{i}}{S}$.

For $\tu v\in\digitSetDim^{*}$ and $\tu w\in\digitSetDim^{\omega}$, we
define $\wordToNatural{\tu{v}}$ as
$\left(\wordToNatural{v_{0}},\dots,\wordToNatural{v_{d-1}}\right)$ and
$\wordToFractional{\tu{w}}$ as
$
\left(\wordToFractional{w_{0}},\dots,\wordToFractional{w_{d-1}}\right)$.
Similarly, we define $\wordToReal{\tu{v\realDot w}}$ as
$\left(\wordToReal{v_{0}\realDot
    w_{0}},\dots,\wordToReal{v_{d-1}\realDot w_{d-1}}\right)=
\wordToNatural{\tu{v}}+\wordToFractional{\tu{w}}$
where addition is defined component by component.  For example:
\begin{equation}\label{eq:ex-par}
  \wordToReal{\left({1\choose 0}{0\choose 0}{1\choose
        0}^{\omega},\set2\right)}=  \wordToNatural{{1\choose 0}{0\choose 0}}+\wordToFractional{{1\choose
      0}^{\omega}}=\left(2,0\right)+\left(\frac23,\frac23\right)=\left(\frac83,\frac23\right).
\end{equation}
\paragraph{Sequential encodings}\label{sec:seq-enc}
We now introduce the notion of sequential encodings of a $d$-vector of
numbers.  Let $w\in\digitSet^{\infty}$ whose length is either a
multiple of $d$ or infinite. Let $\paral{w}\in\digitSetDim$ be the
only word of length $\length{w}/d$, whose $i$-th letter is
$\left(w_{di+0},\dots,w_{di+(d-1)}\right)$ for $i<\length{w}/d$.  For
$\pair{w}{S}$, with $S$ a set of multiple of $d$, let
$\paral{\pair{w}{S}}=\pair{\paral{w}}{\set{s/d\mid s\in
    S}}$\index{par@$\paral{\pair{w}{S}}$}.
Let $\seqen{\tu w}$\index{seq@$\seqen{\tu w}$} be the inverse of the
function $\paral{w}$. Given a word $w$, $\paral{w}$ is called the
\emph{parallelization of $w$}\index{Parallelization of a word} and
$\seqen{w}$ is called its
\emph{sequentialization}\index{Sequentialization of a word}.  For
example, the parallelization of
$\left(100(01)^{\omega},\set{4}\right)$ is
$\left({1\choose 0}{0\choose 0}{1\choose 0}^{\omega},\set2\right)$.
As seen in Equation \eqref{eq:ex-par}, it encodes the pair of reals
$(8/3, 2/3)$.
\subsection{Encoding of sets of vectors of reals}\label{subsec:enc-set-reals}
We now explain how to encode sets of tuples of reals as a language. 
\paragraph{$d$-parallel languages}
The subsets of $\digitSetDim^{*}\realDot\digitSetDim^{\omega}$ are
called \emph{$d$-parallel language}\index{$d$-parallel language}.
Given a $d$-parallel language $L$, let \index{L@$\wordToReal{L}$ for
  $L$ a $d$-parallel language.}$\wordToReal{L}$ be the set of vectors
of reals admitting an encoding in $L$. Formally,
$\wordToReal L=\set{\wordToReal{\tu w}\mid \tu w\in L}$.  The language
$L$ is said to be a $d$-parallel encoding of the set of reals
$\wordToReal L$.  A $d$-parallel language
$L\subseteq\digitSetDim^{*}\realDot\digitSetDim^{\omega}$ is said to
be \emph{saturated}\index{Saturated language} if, for any $d$-vector
of numbers $\mathbf{r}\in\wordToReal{L}$, all encodings in base $b$ of
$\mathbf{r}$ belongs to $L$.



In general, a set of reals may have infinitely many encodings in base
$b$.  For example, for $I\subseteq\mathbb N$ an arbitrary set, the
languagse
$\set{0,1}^{*}\realDot\left(\set{0,1}^{\omega}\setminus\set{0^{i}1^{\omega}\mid
    i\in I}\right)$ is an encoding in base 2 of $\mathbb R^{\ge0}$. It
is saturated only for $I=\emptyset$.
\paragraph{$d$-sequential languages} The case of $d$-sequential
encodings of vectors of reals is now considered. The subsets of
$\digitSet^{d\N}\realDot\digitSet^{\omega}$ are called
\emph{$d$-sequential languages}\index{$d$-sequential language}. The
parallelization of a $d$-sequential language $L$, denoted
$\paral{L}$\index{parL@$\paral{L}$ for $L$ a $d$-sequential language}
is $\set{\paral{w}\mid w\in L}$.  A $d$-sequential language $L$ is
said to be a $d$-sequential encoding of the set
$\wordToReal{\paral{L}}$. This language is said to be saturated if
$\paral{L}$ is saturated.
\section{Deterministic Büchi automata}\label{sec:automata}
This paper deals with deterministic Büchi automata. We define this
notion in \autoref{sec:def-Büchi}. We consider  the  notion of quotient and
morphism of Büchi automata in \autoref{sec:quotient}.
\subsection{Definition}\label{sec:def-Büchi}
A \index{Deterministic Büchi automaton}\emph{deterministic Büchi
  automaton} is a 5-tuple
$\autPar{Q}{\alphabet}{\delta}{\iniState}{F}$, with $Q$ a finite set
of states, $\alphabet$ an alphabet, $\delta: Q\otimes\alphabet\to{}Q$
is the \emph{transition function}, $q_0\in Q$ is the \emph{initial
  state} and $F\subseteq Q$ is the set of \emph{accepting states}. For
each $q\in Q$ and $a\in\alphabet$, $q$ is said to be a
\emph{predecessor} \index{Predecessor state} of $\del{q}{a}$. 
For $q\in Q$, let $\changeIniState{\Aut}{q}$ \index{A@$\changeIniState{\Aut}{q}$ for $\Aut$ an
  automaton and $q$ a state.} be the automaton
$\autPar{Q}{\alphabet}{\delta}{q}{F}$, that is $\Aut$ with $q$ as
initial state. A state $q\in Q$ is said to be
\emph{accessible}\index{Accessible} from a state $q'\in Q$ if there
exists a finite non-empty word $w\in\alphabet^{+}$ such that
$\del{q'}{w}={q}$. The \index{Strongly connected
  components}\emph{strongly connected component} of a state $q$ is the
set of states $q'$ such that $q'$ is accessible from $q$ and $q$ is
accessible from $q'$.

From now on in this paper, all Büchi automata are assumed to be
deterministic. The function $\delta$ is implicitly extended on
$Q\otimes\alphabet^{*}$ by $\del{q}{\epsilon}=q$ and
$\del{q}{aw}=\del{\del{q}{a}}{w}$ for $a\in\alphabet$ and
$w\in\alphabet^{*}$. An example of Büchi automaton is now given.
\begin{example}\label{ex:unbounded}
  Let $\Aut$ be the automaton pictured in \autoref{fig:false-neg}. Its
  alphabet is $\digitDotSet[3]$.
  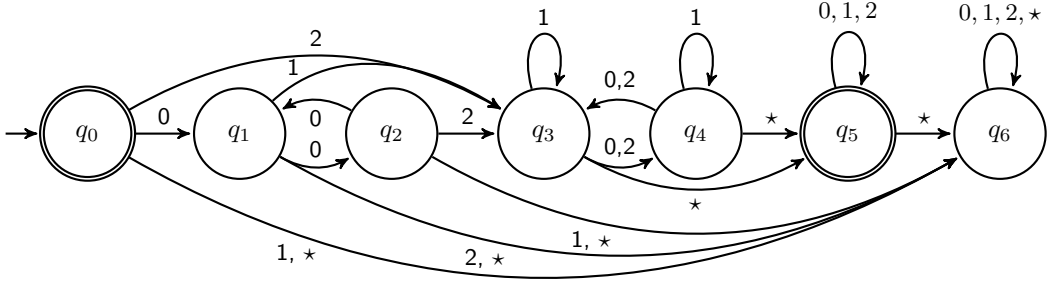
\begin{figure}[h]
  \centering
    \begin{tikzpicture}[->, >=stealth', shorten >=1pt, auto, node
      distance=2cm, thick, main node/.style={circle, fill=!20, draw,
        font=\sffamily\Large\bfseries, inner sep=0.1pt, minimum
        size=1cm}] \tikzset{every state/.style={minimum size=1.2cm}}
      \node[state, initial, accepting, initial text={}] (init) {$\iniState$};
      \node[state, right of=init] (even) {$q_{1}$};
      \node[state, right of=even] (odd) {$q_{2}$};
      \node[state, right of=odd] (acc) {$q_{3}$};
      \node[state, right of=acc] (acc') {$q_{4}$};
      \node[state, right of= acc', accepting] (01) {$q_{5}$};
      \node[state, right of= 01] (empty) {$q_{6}$};

      \path[every node/.style={font=\sffamily\small}]
      (init) edge node {0} (even)
             edge [bend left] node {2} (acc)
             edge [bend right] node [below, pos=.2 ] {1, $\realDot$} (empty)
      (even) edge [bend right] node  {0} (odd)
             edge [bend left, out=45] node [pos=.1,above]{1} (acc)
             edge [bend right] node [below,pos=.3] {2, $\realDot$} (empty)
      (odd)  edge [bend right] node {0} (even)
             edge [bend right] node [below,pos=.3] {1, $\realDot$} (empty)
             edge node [above] {2} (acc)
      (acc) edge [loop above] node {1} (acc)
            edge [bend right] node [below] {$\realDot$} (01)
            edge [bend right] node {0,2} (acc')
      (acc') edge [loop above] node {1} (acc)
            edge node {$\realDot$} (01)
            edge [bend right] node [above] {0,2} (acc)
      (01) edge [loop above] node {$0,1,2$} (01)
           edge node{$\realDot$} (empty)
      (empty) edge [loop above] node{$0,1,2,\realDot$} (empty)
           ;
    \end{tikzpicture}
    \caption{An automaton which recognizes $(00)(2+01)\digitSet[3]^{*}\realDot\digitSet[3]^{\omega}$.
    }
    \label{fig:false-neg}
  \end{figure}
\end{example}
\index{Run of an automaton} Let $\mathcal{A}$ be an automaton and $w$
be an infinite word. A \emph{run}\index{Run} $\pi$ of $\mathcal{A}$ on
$w$ is a mapping $\pi:\N\mapsto Q$ such that $\pi(0)=\iniState$ and
$\del{\pi(i)}{w[i]}={\pi(i+1)}$ for all $i<\length{w}$. The run is
accepting if there exists a state $q\in F$ such that there is an
infinite number of $i\in\N$ such that $\pi(i)=q$.
\autoref{ex:unbounded} is now resumed. Note that, if $\Aut$ is an
automaton, for all $w\in\alphabet^{*}$ and $w'\in\alphabet^{\omega}$,
the word $w'$ is accepted by $\Aut_{\del{q_{0}}{w}}$ if and only if
$ww'$ is accepted by $\Aut$.  It is said that $\Aut$ recognize the
language of words $w$ such that $\Aut$ accepts $w$. This language is
denoted \index{A@$\toInfWord{\Aut}$}$\toInfWord{\Aut}$.
\begin{example}
  Let $\Aut$ be the automaton pictured in \autoref{fig:false-neg}.
  The run of $\Aut$ on $01^{\omega}$ is
  $\tuple{q_{0},q_{1},q_{3},\dots}$, with the last state repeated
  infinitely often.  The Büchi automaton $\Aut$ does not accept
  $01^{\omega}$ since this run contains exactly one accepting state.

  The run of $\Aut$ on $2\realDot1^{\omega}$ is
  $\tuple{q_{0},q_{3},q_{5},\dots}$ with the last state repeated
  infinitely often. The Büchi automaton $\Aut$ accepts
  $2\realDot1^{\omega}$ since the accepting state $q_{5}$ appears
  infinitely often in the run.

  This automaton recognizes the language
  $(00)^{*}(01,2)\digitSet[3]^{*}\realDot\digitSet[3]^{\omega}$.
\end{example}
\subsection{Quotients, Morphisms and Weak Büchi Automata}\label{sec:quotient}
\index{Morphism of automata} \index{Quotient of automata} Let
$\Aut=\autPar{Q}{\alphabet}{\delta}{\iniState}{F}$ and
$\Aut'=\autPar{Q'}{\alphabet}{\delta'}{\iniState'}{F'}$ be two Büchi
automata.  The Büchi automaton $\Aut$ is said to be minimal if, for
each distinct states $q$ and $q'$ of $\Aut$,
$\toInfWord{\changeIniState{\Aut}{q}}\ne\toInfWord{\changeIniState{\Aut}{q'}}$. If $\Aut'$ is minimal,
a surjective function $\mu:Q\to Q'$ is a
\emph{morphism}\index{Morphism of automata} of Büchi automata if
$\mu(\iniState)=\iniState'$ and if, for each $q\in Q$,
$\toInfWord{\changeIniState{\Aut}{q}}=\toInfWord{\changeIniState{Aut'}{\mu(q)}}$. Note that if $\mu$
is a morphism of Büchi automaton from $\Aut$ to $\Aut'$, if $q$ is a
state of $\Aut$, then $\mu$ is a morphism of Büchi automaton from
$\changeIniState{\Aut}{q}$ to $\changeIniState{Aut'}{\mu(q)}$.

The Büchi automaton $\Aut$ is said to be \emph{weak} if $F$ is a union
of strongly connected components.  The main theorem concerning
quotient of weak Büchi automata is now recalled.
\begin{theorem}[\cite{minimal-buchi}]\label{theo:minimal}
  Let $\Aut$ be a weak Büchi automaton with $n$ states such
  that all states of $\Aut$ are accessible from its initial
  state. Let $c$ be the cardinality of $\alphabet$.  There exists a
  minimal weak Büchi automaton $\Aut'$ such that there exists a
  morphism of automaton $\mu$ from $\Aut$ to
  $\Aut'$. The automaton $\Aut'$ and the morphism $\mu$
  are computable in time $\bigO{n\log(n)c}$ and space $\bigO{nc}$.
\end{theorem}
Example \ref{ex:unbounded} is now resumed.
\begin{example}\label{ex:unbounded-min}
  Let $\AR$ be the Büchi automaton pictured in
  \autoref{fig:false-neg}. The automaton $\AR$ is weak. Its minimal
  quotient is pictured in \autoref{fig:false-neg-min}. Note that this
  quotient is not a quotient of finite automata since the accepting
  state $q_{0}$ is sent to a non-accepting state.
\begin{figure}[h]
  \centering
    \begin{tikzpicture}[->, >=stealth', shorten >=1pt, auto, node
      distance=2cm, thick, main node/.style={circle, fill=!20, draw,
        font=\sffamily\Large\bfseries, inner sep=0.1pt, minimum
        size=1cm}] \tikzset{every state/.style={minimum size=1.2cm}}
      \node[state, initial, initial text={}] (even) {$q_{0},q_{2}$};
      \node[state, right of=even] (odd) {$q_{1}$};
      \node[state, right of=odd] (acc) {$q_{3},q_{4}$};
      \node[state, right of=acc, accepting] (01) {$q_{5}$};
      \node[state, right of=01] (empty) {$q_{6}$};

      \path[every node/.style={font=\sffamily\small}]
      (even)  edge [bend right] node {0} (odd)
             edge [bend left, out=45] node {2} (acc)
             edge [bend right] node [below] {1, $\realDot$} (empty)
      (odd)  edge [bend right] node {0} (even)
             edge [bend right] node [below,pos=.3] {2, $\realDot$} (empty)
            edge  node {1} (acc)
      (acc) edge [loop above] node {0,1,2} (acc)
            edge node {$\realDot$} (01)
      (01) edge [loop above] node {$0,1,2$} (01)
           edge  node {$\realDot$} (empty)
      (empty) edge [loop above] node{$0,1,2,\realDot$} (empty)
           ;
    \end{tikzpicture}
    \caption{The minimal automaton which recognizes $(00)(2+01)\digitSet[3]^{*}\realDot\digitSet[3]^{\omega}$.
    }
    \label{fig:false-neg-min}
  \end{figure}
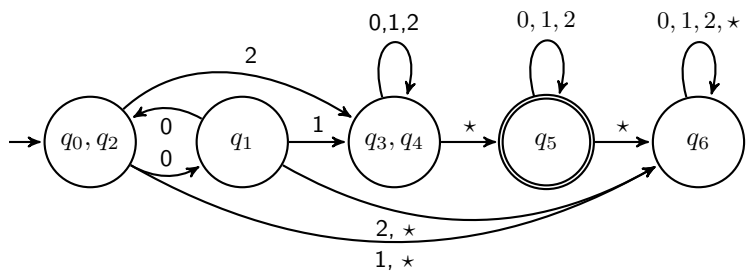
\end{example}
We now explain how to decide efficiently whether two states of two
automata recognize the same language.
\begin{corollary}\label{theo:min-quasi}
  Let $\alphabet$ an alphabet with $c>1$ letters.  Let
  $\Aut^{0}=\autPar{Q^{0}}{\alphabet}{\delta^{0}}{q^{0}_{0}}{F^{0}}$,
  $\Aut^{1}=\autPar{Q^{1}}{\alphabet}{\delta^{1}}{q^{1}_{0}}{F^{1}}$
  be weak Büchi automata.  Let $n=\card{Q^{0}}+\card{Q^{1}}$.

  We can compute in time $\bigO{n\log(n)c}$ and space $\bigO{nc}$ a
  data-structure of size $\bigO{nc}$ such that, for each pair
  $(q^{0},q^{1})\in Q^{0}\otimes Q^{1}$, we can check in constant time
  and space whether $\Aut^{0}_{q^{0}}$ and $\Aut^{1}_{q^{1}}$ accepts
  the same language.
\end{corollary}
\begin{proof}
  Up to changing the name of the states, we can assume that $Q^{0}$
  and $Q^{1}$ are disjoint. Let $\alpha\in\alphabet$. Let
  $\Aut'=\autPar{Q^{0}\cup
    Q^{1}\cup\set\iniState}{\alphabet}{\delta'}{\iniState'}{F^{0}\cup
    F^{1}}$, where $\delta'(\iniState,\alpha)=\iniState^{0}$,
  $\delta'(\iniState,a)=\iniState^{1}$ for each
  $a\in\alphabet\setminus\set a$, and $\delta'(q,a)=\delta^{i}(q,a)$
  for $i\in\set{0,1}$, $a\in\alphabet$ and $q\in Q^{i}$.  Clearly
  $\Aut^{i}_{q}$ accepts the same language than $\Aut'_{q}$, for each
  $q\in Q^{i}$. The automaton $\Aut'$ is clearly weak, thus it admits
  a minimal quotient and a morphism $\mu$ to this minimal quotient. By
  \autoref{theo:minimal}, this morphism is computable in time
  $\bigO{n\log(n)c}$ and takes space $\bigO{nc}$. This morphism is the
  data structure mentionned above.

  Let $(q^{0},q^{1})\in Q^{0}\otimes Q^{1}$. Remark that
  $\Aut^{0}_{q^{0}}$ and $\Aut^{1}_{q^{1}}$ accepts the same language
  if and only if $\mu(q^{0})=\mu(q^{1})$. Given $\mu,$ this equality
  can be checked in constant time and space.
\end{proof}
In practice, the algorithm of \cite{minimal-buchi} could be directly
applied to multiple Büchi automata simultaneously.  Indeed, the
initial state is not considered differently than any other state in
this algorithm. Furtherrmore, this algorithm does not require all
states of the automata to be accessible from the initial state. 
\section{Time and space analysis}\label{sec:time-space}
We now state our assumption above the time and space complexity of in
this paper. We first consider the size of the object we use.

All integers and Booleans takes constant space. The size of an automata is the
product of the cardinalities of its alphabet and of its set of
states. An array of $n$ elements takes size $n$, plus the size of its
elements. For a set $S$ of cardinality $n$, a subset of $S$ is an
array of $n$ Boolean values.

For the sake of simplicity, it is assumed that all basic arithmetic
operations over integers, such as addition, multiplication,
subtraction, comparison of integers, can be computed in constant time
and space. The transition functions of automata return in constant
time and space.  Creating an array and editing one of its position
takes constant time.
\section{Automata reading set of vectors of reals}\label{sec:d-seq-par}
We consider automata reading set of vectors of reals in this
section. Those automata are formally introduced in
\autoref{sec:d-aut}. We explain in \autoref{sec:paral-aut} how to
decide whether an automaton accept a $d$-parallel or a $d$-sequential
language.
\subsection{Definition}\label{sec:d-aut}
The notion of Büchi automata recognizing a set of vector of reals is
now introduced.

A Büchi automaton accepting a $d$-parallel or a $d$-sequential
language is said to be a $d$-parallel or a $d$-sequential automaton respectively.
The set of weak $d$-parallel and of weak $d$-sequential Büchi automata
are closed under taking quotient.  The set of $d$-vectors of automata
associated to an automaton is now introduced.
\begin{notation}[$\wordToReal{\Aut}$]
  \index{A-real@$\wordToReal{\Aut}$} For $\Aut$ a $d$-parallel or a
  $d$-sequential automaton, let $\wordToReal{\Aut}$ be
  $\wordToReal{\toInfWord{\Aut}}$.
\end{notation}
The following example show that the minimal $d$-sequential automaton
accepting a set $R\subseteq\left(\R^{\ge0}\right)^{d}$ can be
exponentially smaller than the minimal $d$-parallel automaton
accepting it.
\begin{example}
  \label{ex:exp-smaller}
  The minimal $d$-parallel automaton accepting $(\R^{\ge0})^{d}$ is:
  \begin{equation*}
    \Aut^{\text{par}}=\autPar{\set{\inftyState,\zuState,\emptyState}}{\digitDotSetDim[2][d]}{\delta}{\inftyState}{\set{\zuState}},
  \end{equation*}
  where $\del{q}{\tu a}=q$ for each state $q$, and each letter
  $\tu a\in \digitSetDim[2][d]$ and where
  $\del{\inftyState}{\realDot}=\zuState$. If $\del{q}{a}$ is not
  defined above, it is equal to $\emptyState$.
  This Büchi automaton has $3$ states and its alphabet has $2^{d}+1$
  letters, hence its size is $\bigO{2^{d}}$.  The automaton
  $\Aut^{\text{par}}$ is pictured in \autoref{ex:min-eq-par}, without
  its state $\emptyState$.

  The minimal $d$-sequential Büchi automaton accepting
  $\left(\R^{\ge0}\right)^{d}$ is:
  \begin{equation*}
    \Aut^{\text{seq}}=\autPar{\set{\emptyState,\zuState}\cup\set{q_{i}\mid{i\in[d-1]}}}{\digitDotSet[2]}{\delta}{q_{0}}{\set{\zuState}},
  \end{equation*}
  where, $\del{q_{i}}{a}=q_{i+1}$ for each $a\in\digitSet[2]$, where
  $\del{\zuState}{a}=\zuState$ for each $a\in\digitSet[2]$ and where
  $\del{\iniState}{\realDot}=\zuState$. If $\del{q}{a}$ is not defined
  above, it is equal to $\emptyState$.
  This Büchi automaton has $d+2$ states and its alphabet has $2$
  letters, hence its size is $\bigO{d}$.  The automaton
  $\Aut^{\text{par}}$ is pictured in \autoref{ex:min-eq-seq} without
  its state $\emptyState$.
  \begin{figure}[h]
    \begin{subfigure}[b]{0.25\textwidth}
      \centering
      \begin{tikzpicture}[->, >=stealth', shorten >=1pt, auto, node
        distance=2cm, thick, main node/.style={circle, fill=!20, draw,
          font=\sffamily\Large\bfseries, inner sep=0.1pt, minimum
          size=1cm}] \tikzset{every state/.style={minimum size=1.2cm}}
        \node[state, accepting, initial,initial text={}] (init) {$\iniState$};
        \node[state, accepting,right of=init] (zu) {$\zuState$};
        
        \path[every node/.style={font=\sffamily\small}]
        (init) edge [loop above] node {$\digitSet^{d}$} (init)
               edge node {$\realDot$} (zu)
        (zu) edge [loop above] node {$\digitSet^{d}$} (zu)
        ;
      \end{tikzpicture}
      \caption{The minimal parallel Büchi automaton.}
      \label{ex:min-eq-par}
    \end{subfigure}
    \begin{subfigure}[b]{0.7\textwidth}
      \centering
      \begin{tikzpicture}[->, >=stealth', shorten >=1pt, auto, node
        distance=2cm, thick, main node/.style={circle, fill=!20, draw,
          font=\sffamily\Large\bfseries, inner sep=0.1pt, minimum
          size=1cm}] \tikzset{every state/.style={minimum size=1.2cm}}
        \node[state, initial, accepting, initial text={}] (init) {$\iniState$};
        \node[state, accepting, right of=init] (q1) {$q_{1}$};
        \node[right of=q1] (q2) {\dots};
        \node[state,accepting,right of=q2] (q3) {$q_{d-1}$};
        \node[state, accepting,right of=q3] (zu) {$\zuState$};

        \path[every node/.style={font=\sffamily\small}]
        (init) edge node {$\digitSet$} (q1)
               edge [bend left] node [below] {$\realDot$} (zu)
        (q1) edge node {$\digitSet$} (q2)
        (q2) edge node {$\digitSet$} (q3)
        (q3) edge [bend left] node {$\digitSet$} (init)
        (zu) edge [loop above] node {$\digitSet$} (zu)
        ;
      \end{tikzpicture}
      \caption{The minimal sequential Büchi automaton.}
      \label{ex:min-eq-seq}
    \end{subfigure}
    \caption{Minimal parallel and sequential RVA accepting
      $\left(\R^{\ge0}\right)^{d}$.}
  \end{figure}
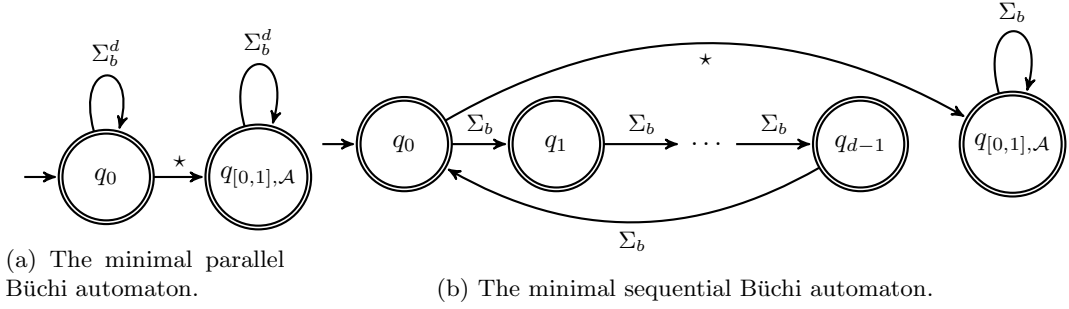
  Note that the size of the minimal $d$-parallel automaton is
  exponential in the size of the minimal $d$-sequential automaton.
\end{example}
We now explain how to transform a sequential automaton into a parallel
one.
\begin{definition}[$\paral{\Aut}$]
  Let $\Aut=\autPar{Q}{\digitDotSet}{\delta}{q_{0}}{F}$ be a
  $d$-sequential automaton. Let
  $\paral{\delta}:\left(Q\otimes\left(\digitDotSetDim\right)\right)\to
  Q$ such that
  $\paral{\delta}\left(q,\realDot\right)=\del{q}{\realDot}$ and such
  that,
  $\paral{\delta}\left(q,\tu a\right)=\delta(q,a_{0}\dots a_{d-1})$
  for each $\tu a\in\digitSetDim$. Then let $
  \paral{\Aut}=\autPar{Q}{\digitDotSetDim}{\paral{\delta}}{\iniState}{F}$.
\end{definition}
This operation is called the \emph{parallelization of
  $\Aut$}\index{Parallelization of an automaton}. 
The parallelization of the automaton pictured in
\autoref{ex:min-eq-seq} is pictured in \autoref{ex:min-eq-seq-par}.
  \begin{figure}
    \centering
    \begin{tikzpicture}[->, >=stealth', shorten >=1pt, auto, node
      distance=3cm, thick, main node/.style={circle, fill=!20, draw,
        font=\sffamily\Large\bfseries, inner sep=0.1pt, minimum
        size=1cm}] \tikzset{every state/.style={minimum size=1.2cm}}
      \node[state, initial, accepting, initial text={}] (init) {$\iniState$};
      \node[state, accepting, right of=init] (q1) {$q_{1}$};
      \node[right of= q1] (q2) {\dots};
      \node[state,accepting,right= 1cm of q2] (q3) {$q_{d-1}$};
      \node[state, accepting,right of=q3] (zu) {$\zuState$};
       
      \path[every node/.style={font=\sffamily\small}]
      (init) edge [loop right] node {$\digitSetDim$} (q1)
             edge [bend left,out=20, in=160] node [below] {$\realDot$} (zu)
      (q1) edge [loop right] node {$\digitSetDim$} (q1)
      (q3) edge [loop right] node {$\digitSetDim$} (q3)
      (zu) edge [loop right] node {$\digitSetDim$} (zu)
      ;
    \end{tikzpicture}
    \caption{The parallelization of the automaton of \autoref{ex:min-eq-seq}.}
    \label{ex:min-eq-seq-par}
  \end{figure}
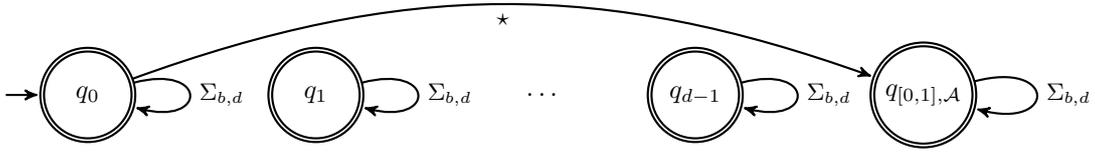
We now state two lemmas whose proofs are simple applications of the
definitions. Those lemmas show that this notion of parallelization is
coherent with the parallelization of words.
\begin{lemma}\label{lem:paral-aut}
  Let $w\in\digitSet^{d\N}\realDot\digitSet^{\omega}$ and $\Aut$
  a $d$-sequential automaton. The automaton $\Aut$ accepts $w$ if and
  only if $\paral \Aut$ accepts $\paral w$.
\end{lemma}
\begin{lemma}\label{lem:seq-aut}
  Let $\tu w\in\digitSetDim^{*}\realDot\digitSetDim^{\omega}$ and
  $\Aut$ a $d$-sequential automaton. The automaton $\paral{\Aut}$
  accepts $\tu w$ if and only if $\Aut$ accepts $\seqen w$.
\end{lemma}
Finally, we state that changing the initial state commute with
parallelization.
\begin{lemma}\label{lem:paral-state}
  Let $\Aut$ be a $d$-sequential automaton and $q$ a state of
  $\Aut$. The automaton $\paral{\Aut}_{q}$ is equal to the automaton
  $\paral{\changeIniState{\Aut}{q}}$.
\end{lemma}

\subsection{Algorithm}\label{sec:paral-aut}
We now consider the problem of deciding whether a Büchi automaton is
$d$-parallel or $d$-sequential. We now state the two main results of
this section.
\begin{theorem}\label{theo:d-par}
  Let $\Aut$ be an automaton over alphabet $\digitDotSetDim$ with $n$
  states. It is decidable in time $\bigO{nb^{d}}$ and space $\bigO{n}$
  whether $\Aut$ is a $d$-parallel automaton.
\end{theorem}
\begin{theorem}\label{theo:d-seq}
  Let $\Aut$ be an automaton over alphabet $\digitDotSet$ with $n$
  states. It is decidable in time $\bigO{nbd}$ and space $\bigO{nd}$
  whether $\Aut$ is a $d$-sequential automaton.
\end{theorem}
We prove both theorems  simultaneously. More precisely, we prove the
following proposition:
\begin{proposition}\label{prop:d-seq-par}
  Let $d_{\text{par}},d_{\text{seq}}>0$.  Let $\Aut$ be an automaton over alphabet
  $\digitDotSetDim[b][d_{\text{par}}]$ with $n$ states. It is decidable in time
  $\bigO{nb^{d_{\text{par}}}d_{\text{seq}}}$ and space $\bigO{nd_{\text{seq}}}$ whether $\Aut$
  recognize a subset of
  $\digitSetDim[b][d_{\text{par}}]^{d_{\text{seq}}\N}\realDot\digitSetDim[b][d_{\text{par}}]^{d-1}$.
\end{proposition}
When $d_{\text{par}}=d$ and $d_{\text{seq}}=1$, this proposition implies
\autoref{theo:d-par}, and when $d_{\text{par}}=1$ and $d_{\text{seq}}=d$, this
proposition implies \autoref{theo:d-seq}.  In order to prove this
proposition, we introduce the following sets of states.
\begin{definition}\label{def:states}
  [$\emptyStates$, $\fraStates$ and
  $\modStates{i}$]\label{not:set-states} Let $\Aut$ be an automaton
  over alphabet $\digitDotSetDim[b][d_{\text{par}}]$.
  \begin{itemize}
  \item Let $\emptyStates$ \index{QO@$\emptyStates$} be the set of states
    $q$ such that $\mathcal{A}_{q}$ recognizes the empty language.
  \item For $i\in[d_{\text{seq}}-1]$, let $\modStates{i}$ be the set
    of states $\del{q}{\tu w}$ with
    $\tu w\in\digitSetDim[b][d_{\text{par}}]^{d_{\text{seq}}\N+i}$.
  \item Let $\fraStates$\index{Qfra@$\fraStates$} be the set of states
    $\del{q}{\tu w}$ with
    $\tu
    w\in\digitSetDim[b][d_{\text{par}}]^{*}\realDot\digitSetDim[b][d_{\text{par}}]^*$
  \end{itemize}
\end{definition}
Intuitively, while the automaton reads the natural part of a vector,
it visits successively a state of $\modStates 0$, a state of
$\modStates1$, \dots, a state of $\modStates d-1$, and then, again a
state of $\modStates 0$ and so on. Similarly, while the automaton read
the fractional part of the word, it visits states of $\fraStates$. We
could prove that, if $\Aut$ accepts a subset of
$\digitSetDim[b][d_{\text{par}}]^{d_{\text{seq}}\N}\realDot\digitSetDim[b][d_{\text{par}}]^{\omega}$,
then the intersection of two of those sets is included in
$\emptyStates$.  We now give example of those set of states.
\begin{example}
  Let $\Aut$ be the automaton pictured in \autoref{ex:min-eq-par},
  $d_{\text{par}}=d$ and $d_{\text{seq}}=1$. Then
  $\emptyStates=\set\emptyState$, $\modStates{0}=\set{\iniState}$ and
  $\fraStates=\set{\zuState}$.

  Let $\Aut$ be the automaton pictured in \autoref{ex:min-eq-seq},
  $d_{\text{par}}=1$ and $d_{\text{seq}}=2$. Then
  $\emptyStates=\set\emptyState$, $\modStates{i}=\set{q_{i}}$ and
  $\fraStates=\set{\zuState,\emptyState}$.
\end{example}
We now characterize the automata accepting a subset of
$\digitSetDim[b][d_{\text{par}}]^{d_{\text{seq}}\N}\realDot\digitSetDim[b][d_{\text{par}}]^\omega$
using sets introduced in \autoref{def:states}. We then characterize
those sets of states. All characterizations of those objects are
straightforward from their definitions.
\begin{proposition}\label{lem:char-seq-par}
  Let $\Aut$ be an automaton over alphabet
  $\digitDotSetDim[b][d_{\text{par}}]$. It accepts a subset of
  $\digitSetDim[b][d_{\text{par}}]^{d_{\text{seq}}\N}\realDot\digitSetDim[b][d_{\text{par}}]^{d-1}$
  if and only if, $\del{q}{\realDot}\in\emptyStates$ for each state
  $q\in\fraStates\cup\bigcup_{i=1}^{d-1}\modStates{i}$.
\end{proposition}
\begin{lemma}\label{lem:char-empty}
  The set $\emptyStates$ is the greatest set of states included in
  $Q$, which does not contain the accepting recurrent states and which
  is closed under taking predecessor.
\end{lemma}
\begin{lemma}\label{lem:char-mod}
  The sets $\modStates{0},\dots,\modStates{d_{\text{seq}}-1}$ is the
  smallest family of sets such that $q_{0}\in\modStates{0}$ and for
  each $i\in[d_{\text{seq}}-1]$, for each $q\in\modStates{i}$ and for
  each $\tu a\in\digitSetDim[b][d_{\text{par}}]$, the state
  $\del{q}{\tu a}$ belongs to $\modStates{i+1}$.
\end{lemma}
\begin{lemma}\label{lem:char-fra}
  The set $\fraStates$ is the smallest set containing all sets of the
  form $\del{q}{\realDot}$ for $q\in\bigcup_{i=0}^{d-1}\modStates{i}$
  and $\del{q}{\tu a}$ for $q\in\fraStates$ and
  $\tu a\in\digitSetDim[b][d_{\text{par}}]$.
\end{lemma}
It is now explained how to compute efficiently those sets.
\begin{lemma}\label{lem:alg:sets}
  All sets of \autoref{def:states} are computable in time
  $\bigO{nb^{d_{\text{par}}}d_{\text{seq}}}$ and space
  $\bigO{nd_{\text{seq}}}$.
\end{lemma}
\begin{proof}
  Let us first consider the set $\emptyStates$.  The algorithms is a
  straightforward application of the characterization given in Lemma
  \ref{lem:char-empty}.  Tarjan's algorithm \cite{Tarjan} can be used
  to compute the set of strongly connected component in time
  $\bigO{nb^{d_{\text{par}}}}$, and therefore the set of recurrent
  states. Furthermore, it is easy to associate in linear time to each
  state its set of predecessors. Let $p_{q}$ be the number of
  predecessors of a state $q$.
  
  Two sets $\texttt{PotentiallyEmpty}$ and $\texttt{ToProcess}$ are
  used by the algorithm.  The algorithm initializes the set
  $\texttt{PotentiallyEmpty}$ to $Q$ and initializes the set
  $\texttt{ToProcess}$ to the empty set. The algorithm runs on each
  recurrent state $q$. For each state $q$, if $q$ is accepting, then
  $q$ is removed from $\texttt{PotentiallyEmpty}$ and added to
  $\texttt{ToProcess}$. The algorithm then runs on each element $q$ of
  $\texttt{ToProcess}$. For each state $q$, the algorithm removes $q$
  from $\texttt{ToProcess}$ and runs on each predecessors $q'$ of $q$.
  For each $q'$, if $q'$ is in $\texttt{PotentiallyEmpty}$, then $q'$
  is removed from $\texttt{PotentiallyEmpty}$ and added to
  $\texttt{ToProcess}$. Finally, when $\texttt{ToProcess}$ is empty,
  the algorithm halts and $\emptyStates$ is the value of
  $\texttt{PotentiallyEmpty}$.

  Let us now consider the complexity of this algorithm. At most $n$
  states are added to $\texttt{ToProcess}$, and each state is added at
  most once. For each state $q$ added to $\texttt{ToProcess}$, each of
  its $p_{q}$ predecessor is considered in constant time. Thus the
  algorithm runs in time
  $\bigO{n+ \sum_{q\in Q}p_{q}}=\bigO{nb^{d_{\text{par}}}}$.

  \paragraph{}  
  The $d_{\text{seq}}$-sets $\modStates{i}$ are computed
  simultaneously, using the characterization given in
  \autoref{lem:char-mod}.  The computation of the state $\fraStates$
  is similar.  The $d$ sets $\modStates{i}$ are initialized to the
  empty set, and $\texttt{ToProcess}$ is initialized to
  $\set{(\iniState,0)}$.  The algorithm runs on each
  $(q,i)\in\texttt{ToProcess}$. For each $(q,i)$, the pair $(q,i)$ is
  removed from $\texttt{ToProcess}$ and added into
  $\modStates{i}$. The algorithm runs on each
  $\tu a\in\digitSetDim[b][d_{\text{par}}]$. For each $\tu a$, if
  $\del{q}{\tu a}\not\in\modStates{\left(i+1\mod
      d_{\text{seq}}\right)}$, then $\del{q}{\tu a}$ is added into
  $\texttt{ToProcess}$. When $\texttt{ToProcess}$ is empty, all states
  of $\modStates{i}$ are indeed added to this set.

  Let us consider the complexity. Each pair $(q,i)$ is removed
  at most once from $\texttt{ToProcess}$, thus the outer loop is
  executed at moste $\bigO{nd_{\text{seq}}}$ times. Each execution of this loop
  clearly runs in time $\bigO{b^{d_{\text{par}}}}$. This algorithm stores a
  number, a state, and $d_{\text{seq}}$ set of states, thus it clearly takes space
  $\bigO{nd_{\text{seq}}}$.



\end{proof}
We now prove \autoref{prop:d-seq-par}.
\begin{proof}
  The algorithm computes the sets of \autoref{def:states}.  The
  algorithm runs on each
  $q\in\fraStates\cup\bigcup_{i=1}^{d-1}\modStates{i}$ and rejects if
  $\del{q}{\realDot}\not\in\emptyStates$.  By
  \autoref{lem:char-seq-par}, this algorithm accepts if an only if
  $\Aut$ accepts a subset of
  $\digitSetDim[b][d_{\text{par}}]^{d_{\text{seq}}\N}\realDot\digitSetDim[b][d_{\text{par}}]^{d-1}$.
  By \autoref{lem:alg:sets}, those sets can be computed in time
  $\bigO{nb^{d_{\text{par}}}d_{\text{seq}}}$ and space
  $\bigO{nd_{\text{seq}}}$, and the loops clearly runs in time
  $\bigO{n}$ and constant space. Thus, the whole algorithm runs in
  time $\bigO{nb^{d_{\text{par}}}d_{\text{seq}}}$ and space
  $\bigO{nd_{\text{seq}}}$.
\end{proof}
\section{Fixing a component}\label{sec:fix-01}
In order to consider the two encodings of rational numbers, we must
consider the vector of words $\pair{\tu w}{S}$ such that a suffix of some
component $w_{f}$ belongs to $0^{\omega}$ or $(b-1)^{\omega}$.  More
precisely, at some points of the run, the automaton will only have to
$0^{\omega}$ or $(b-1)^{\omega}$ in some component $f$. Since a
component is fixed, we may change the alphabet to fix this letter. We
do it by replacing this $f$-th component by an atomic symbol
$\square$.

We define a function which remove some component of a word in
\autoref{seq:modif-vector}. We introduce the automata which reads
those new words in \autoref{sec:mod-aut-par}.
\subsection{Vector of words}\label{seq:modif-vector}
We now introduced a new alphabet. Letters of this alphabet correspond
to letters of $\digitSetDim$ with some component fixed. We could have
considered $\digitSetDim[b][d-1]$, but this would lead to trouble when
$d=1$. Indeed, a word would then be an element of
$()^{*}\realDot()^{\omega}$, where $()$ is the unique 0 tuple, and it
would not be clear what the sequentialization of such a word would
be. Instead of removing a component, we choose to replace it with an
atomic symbol $\square$. The formal definition is now given.
\begin{definition}[$\squarphabetDim$]
  \label{def:change}
  \index{Sigma@$\squarphabetDim$}\index{wia@$\changePos{f}{\pair{\tu w}{S}}{z}$}
  For $f\in[d-1]$, let
  $\squarphabetDim=\left(\digitSetDim[b][f-1]\otimes\set{\square}\otimes\digitSetDim[d][b-f]\right)$.
\end{definition}
We now introduce a notation to change a word by fixing one of its component.
\begin{definition}[$\changePos{f}{\pair{\tu w}{S}}{z}$] We first
  consider $d$-parallel words.  Let $f\in[d-1]$, $z\in\squarphabet$,
  $\tu w$ a word whose alphabet is a set of $d$-tuples.  Let
  $S\subseteq\mathbb N$. Let $\changePos{f}{\pair{\tu w}{S}}{z}$ be
  $\pair{\tu w'}{S}$, where $\length{\tu w'}=\length{\tu w}$,
  $w'_{f}\in z^{\infty}$, and $w'_{i}=w_{i}$ for each $i\ne f$.

  We now consider $d$-sequential words. Let
  $\pair{w}{S}\in\left(\squarphabetDot\right)^{\infty}$, then
  $\changePos{f}{\pair{w}{S}}{z}$ is
  $\pair{\seqen{\changePos{f}{\paral{w}}{z}}}{S}$. Equivalently, this
  transformation consists in replacing each letter whose position is
  equivalent to $f$ modulo $d$ - not counting the $\realDot$'s - by
  the letter $z$.
\end{definition}
The following three lemmas about $\fix^{z@f}$ are straightforward
consequences from its definition.
\begin{lemma}\label{lem:change-move}
  Let $0\le i<f<d$ integers, $z\in\squarphabet$,
  $v\in\left(\squarphabet\right)^{i}$ and
  $w\in\left(\squarphabetDot\right)^{\infty}$, then
  $\changePos{f}{vw}{z}=v\changePos{f-i}{w}{z}$.
\end{lemma}
Note that the letter $\realDot$ does not belongs to $v$.
%
%
We now state that, if $w_{f}\in z^{\infty}$, then changing twice
position $f$ is equivalent to doing only the last change.
\begin{lemma}\label{lem:change-twice}
  Let $f\in[d-1]$, $z,z'\in\squarphabet$, let $w$ a $d$-parallel or a
  $d$-sequential word.  We have
  $\changePos{f}{\changePos{f}{\pair{\tu{w}}{S}}{z'}}{z}=
  \changePos{f}{\pair{\tu{w}}{S}}{z}$.
\end{lemma}
We now state that, if $w_{f}\in z^{\infty}$, then $\tu w$ is a
fixpoint of the function.
\begin{lemma}\label{lem:change-cancel}
  Let $f\in[d-1]$. Let $z\in\digitSet\cup\set\square$.  Let $\tu w$ be
  a $d$-parallel word.  If $w_{f}\in z^{\infty}$ then
  $\changePos{f}{\pair{\tu{w}}{S}}{z}= \pair{\tu{w}}{S}$.
\end{lemma}

\subsection{Automata}\label{sec:mod-aut-par}
A notation is now introduced, in order to fix the digit read in some
position of an automaton. In this section, we fix $f\in[d-1]$,
$z\in\digitSet$.
\begin{definition}[$\AutFix{f}{z}$]\label{def:autfix-R}
  Let $\Aut=\autPar{Q}{\digitDotSetDim}{\delta}{q_{0}}{F}$ be a
  $d$-parallel Büchi automaton, then let: \index{AJA@$\AutFix{f}{z}$
    for $\Aut$ a $d$-parallel Büchi automaton}
  \begin{equation*}
    \AutFix{f}{z}=\autPar{Q}{\squarphabetDotDim}{\FixTran{f}{z}}{q_{0}}{F},
  \end{equation*}
  where $\FixTran{f}{z}(q,\tu a)=\del{q}{\addPos{f}{\tu a}{z}}$ for
  all $\tu a\in\squarphabetDotDim$.
\end{definition}
\begin{definition}[$\AutFixSeq{z}$]\label{def:autfix-R}
  Let $\Aut=\autPar{Q}{\digitDotSet}{\delta}{q_{0}}{F}$ be a
  $d$-sequential Büchi automaton.  Let: \index{AJA@$\AutFixSeq{z}$ for
    $\Aut$ a $d$-sequential Büchi automaton}
  \begin{equation*}
    \AutFixSeq{z}=\autPar{Q'}{\squarphabetDot}{\FixTranSeq{z}}{(\iniState,0)}{F},
  \end{equation*}
  where $Q'=\left(Q\otimes [d-1]\right)\cup\set\emptyState$ and where,
  for each state $q\in Q$,
  $\FixTranSeq{z}((q,i),\realDot)=(\del{q}{\realDot},i)$,
  $\FixTranSeq{z}((q,i),a)=\left(\del{q}{a},i+1\right)$ for each
  $a\in\digitSet$ and $i\in\set{0,\dots,d-2}$,
  $\FixTranSeq{z}((q,d-1),\square)=\left(\del{q}{z},0\right)$. For
  each state $q\in Q'$ and each letter $a\in\squarphabetDot$ such that
  $\FixTranSeq{z}(q,a)$ is not yet defined, it is set to
  $\emptyState$.
\end{definition}
It is now stated that the two transformations introduced above
preserve weakness.
\begin{lemma}\label{lem:fix-weak}
  Let $\Aut$ be a weak $d$-parallel automaton, then the automaton
  $\AutFix fz$ is weak. Let $\Aut$ be a weak $d$-sequential automaton,
  then the automaton $\AutFixSeq z$ is weak.
\end{lemma}
\begin{proof}
  For the case of $d$-parallel automata, it suffices to remark that
  each strongly connected component of $\AutFix fz$ is a subset of a
  strongly connected component of $\Aut$. For the case of
  $d$-sequential automata, it suffices to remark that a strongly
  connected component of $\AutFixSeq z$ is either $\set\emptyState$,
  or a set $S$ such that $\set{q\mid (q,i)\in S}$ is a strongly
  connected component of $\Aut$.
\end{proof}
We also state that changing the initial state of a $d$-parallel
automaton commute with the transformation introduced above.
\begin{lemma}\label{lem:pemut-fix-state-par}
  Let $\Aut$ be a $d$-parallel automaton and $q$ a state of
  $\Aut$. Then
  $\left(\AutFix fz\right)_{q}=\AutFix[\left(\changeIniState{\Aut}{q}\right)] fz$.
\end{lemma}
\begin{lemma}\label{lem:pemut-fix-state-seq}
  Let $\Aut$ be a $d$-sequential automaton and $q$ a state of
  $\Aut$. Then
  $\left(\AutFixSeq z\right)_{(q,0)}=\AutFixSeq[\left(\changeIniState{\Aut}{q}\right)]
  z$.
\end{lemma}
\paragraph{Words and automata}\label{sec:fix-word-aut}
We show, in this section, how the notations introduced in the two
preceding sections interact. The first two lemmas deal with replacing
a component with a $\square$. 
\begin{lemma}\label{lem:remove-fix-par}
  Let $\Aut$ 
  be a $d$-parallel
  automaton.  Let $\pair{\tu
    w}{S}\in\left(\digitDotSetDim\right)^{\omega}$ such that $w_{f}=
  z^{\omega}$. The automaton $\Aut$ accepts $\pair{\tu
    w}{S}$ if and only if
  $\AutFix{f}{z}$ accepts $\removePos{f}{\pair{\tu w}{S}}$.
\end{lemma}
\begin{lemma}\label{lem:remove-fix-seq}
  Let $\Aut
  $ be a $d$-sequential automaton. 
  Let
  $\pair{w}{S}\in\digitDotSet^{\omega}$ such that
  $\paral{w}_{d-1}=z^{\omega}$.  The automaton
  $\Aut$ accepts $\pair{w}{S}$ if and only if
  $\AutFixSeq{z}$ accepts $\removePos{d-1}{\pair{w}{S}}$.
\end{lemma}
The following lemma deals with replacing a component of a vector by a
single letter.
\begin{lemma}\label{lem:add-fix-par}
  Let $\Aut$ 
  be a $d$-parallel automaton. 
  Let $\pair{\tu w}{S}\in\left(\squarphabetDotDim\right)^{\omega}$. The
  automaton $\AutFix{f}{z}$ accepts $\pair{\tu w}{S}$ if and only if
  $\Aut$ accepts $\addPos{f}{\pair{\tu w}{S}}{z}$.
\end{lemma}
The following lemma illustrates how the notation introduced above
behaves on a word with a component whose suffix is $z^{\omega}$.
\begin{lemma}\label{lem:ult-fix}
  Let $\Aut=\autPar{Q}{\digitDotSet}{\delta}{q_0}{F}$, $q\in Q$,
  $\pair{\tu v}{V}\in\left(\digitDotSetDim\right)^{*}$ and
  $\pair{\tu w}{W}\in\left(\digitDotSetDim\right)^{\omega}$ with
  $w_{f}=z^{\omega}$. Then $\Aut$ accepts
  $\pair{\tu v}{V}\pair{\tu w}{W}$ if and only if
  $\changeIniState{\left(\AutFix f z\right)}{\del{q}{\pair{\tu
        v}{V}}}$ accepts $\removePos{f}{\pair{\tu w}{W}}$.
\end{lemma}
Note that, in the last term, the function $\delta$ is still the
transition function of $\Aut$ and not the one of
$\AutFix{f}{z}$.
\begin{proof}
  The fact \quo{$\Aut$ accepts $\pair{\tu v}{V}\pair{\tu w}{W}$} is
  equivalent to
  \quo{$\changeIniState{\Aut}{\del{\iniState}{\pair{\tu v}{V}}}$
    accepts $\pair{\tu w}{W}$}. By \autoref{lem:remove-fix-par}, since
  $w_{f}=z^{\omega}$, those statements are equivalent to
  \quo{$\AutFix[\left(\changeIniState{\Aut}{\del{\iniState}{\pair{\tu
            v}{V}}}\right)]{f}{z}$
    accepts $\removePos{f}{\pair{\tu w}{W}}$}. By
  \autoref{lem:fix-weak}, they are also equivalent to:
  \quo{$\changeIniState{\left(\AutFix f z\right)}{\del{q}{\pair{\tu
          v}{V}}}$ accepts $\removePos{f}{\pair{\tu w}{W}}$}.
\end{proof}
\section{Characterizations of $d$-parallel
  RVA}\label{sec:char-RVA-par}
Recall from the introduction that a RVA is a Büchi automaton accepting
a saturated language.  In this section, we give some characterizations
of the $d$-parallel RVAs. The last of those characterization allows us
to give in \autoref{sec:algo-RVA} an algorithm which decides whether
an automaton over alphabet $\digitDotSetDim$ is a RVA.  We use the
other characterizations to prove the last one. We first prove a
property of minimal RVA.
\begin{lemma}\label{lem:init-0}
  Let $\Aut=\aut$ a minimal $d$-parallel RVA, then
  $\del{\iniState}{\tu 0}=\iniState$.
\end{lemma}
\begin{proof}
  Since $\Aut$ is minimal, it suffices to prove
  $\toInfWord{\Aut_{\del{\iniState}{\tu 0}}}=\toInfWord{\Aut}$, hence
  it suffices to prove that $\Aut$ accepts $\tu w$ if and only if
  $\Aut_{\del{\iniState}{\tu 0}}$ for all
  $\tu w\in \left(\digitDotSetDim\right)^{\omega}$.

  Let $\tu w\in \left(\digitDotSetDim\right)^{\omega}$.  The automaton
  $\Aut_{\del q{\tu 0}}$ accepts $\tu w$ if and only if $\Aut$ accepts
  $\tu 0\tu w$.  Since $\Aut$ is saturated and
  $\wordToReal{\tu 0\tu w}=\wordToReal{\tu w}$, $\Aut$ accepts
  $\tu 0\tu w$ if and only if it accepts $\tu w$. By transitivity of
  equivalence, $\Aut$ accepts $\tu w$ if and only if
  $\Aut_{\del{\iniState}{\tu 0}}$ accepts $\tu w$.
\end{proof}
We now state the  characterizations.
\begin{proposition}\label{prop:char-RVA}
  Let $\Aut$ be a minimal weak $d$-parallel Büchi automaton such that
  $\del{\iniState}{\tu 0}=\iniState$.  The following statement are
  equivalent:
  \begin{enumerate}
  \item\label{lem:char-RVA-RVA} The automaton $\Aut$ is a RVA.
  \item\label{lem:char-RVA-nat-fixed}For all
    $\pair{\tu w}{\set s}\pair{\tu w'}{\set s}\in
    \digitSetDim^*\realDot\digitSetDim^{\omega}$ such that
    $\wordToReal {\pair{\tu{w}}{\set
        s}}=\wordToReal{\pair{\tu{w'}}{\set s}}$ and such that $\Aut$
    accepts $\pair{\tu w}{\set s}$, $\Aut$ accepts $(\tu w',\set s)$.
  \item\label{lem:char-RVA-dist} for all
    $\pair{\tu w}{\set s}\pair{\tu w'}{\set s}\in
    \digitSetDim^*\realDot\digitSetDim^{\omega}$ such that
    $\wordToReal {\pair{\tu{w}}{\set
        s}}=\wordToReal{\pair{\tu{w'}}{\set s}}$, such that
    $\distWord{ w}{ {w'}}=1$, and such that $\Aut$ accepts
    $\pair{\tu w}{\set s}$, $\Aut$ accepts $(\tu w',\set s)$.
  \item\label{lem:char-RVA-char} for each $f\in[d-1]$, for each $q\in Q$,
    accessible in $\Aut$ from $q_{0}$, for each $\tu a\in\digitSetDim$
    with $a_{f}<b-1$, 
    \begin{equation}\label{eq:char-RVA}
      \toInfWord{\changeIniState{\left(\AutFix{f}{(b-1)}\right)}{\del{q}{\tu{a}}}}=
      \toInfWord{\changeIniState{\left(\AutFix{f}{0}\right)}{\del{q}{\tu a'}}},
    \end{equation} where 
    $\tu a'=a_{0}\dots{}a_{f-1}(a_{f}+1)a_{f+1}\dots{}a_{d-1}$.
  \end{enumerate}
\end{proposition}
Note that Property \eqref{lem:char-RVA-RVA} requires to consider
$d$-tuple of words, with $\realDot$'s in potentially different
positions, while in Property \eqref{lem:char-RVA-nat-fixed}, both words
have a single $\realDot$'s at the same position.  Property
\eqref{lem:char-RVA-nat-fixed} requires to consider any pair of word
whose natural part have the same length, while Property
\eqref{lem:char-RVA-dist} restrict our study to the case where all but
one components of the word are equal.
\begin{proof}
  The proof is done by the following sequence of implications.
  Property \eqref{lem:char-RVA-RVA} implies Property
  \eqref{lem:char-RVA-char}, which implies Property
  \eqref{lem:char-RVA-dist}, which implies Property
  \eqref{lem:char-RVA-nat-fixed}, which implies Property
  \eqref{lem:char-RVA-RVA}.
  \paragraph{Property \eqref{lem:char-RVA-RVA} implies Property
    \eqref{lem:char-RVA-char}}
  Let $f,q,\tu a,\tu a'$ as in the hypothesis. Let us prove that
  $\toInfWord{\changeIniState{\left(\AutFix{f}{(b-1)}\right)}{\del{q}{\tu{a}}}}\subseteq
  \toInfWord{\changeIniState{\left(\AutFix[\left(\Aut\right)]{f}{0}\right)}{\del{q}{\tu{a'}}}}$,
  the proof of the reverse inclusion is similar.  Let
  $\pair{\tu w}{T}\in\left(\squarphabetDotDim\right)^{\omega}$
  accepted by
  $\changeIniState{\left(\AutFix{f}{(b-1)}\right)}{\del{q}{\tu{a}}}$,
  let us prove that it is accepted by
  $\changeIniState{\left(\AutFix[\left(\Aut\right)]{f}{0}\right)}{\del{q}{\tu{a'}}}$.

  Since $\changeIniState{\left(\AutFix{f}{(b-1)}\right)}{\del{q}{\tu{a}}}$
  accepts $\pair{\tu w}{T}$, by Lemma \ref{lem:add-fix-par},
  $\changeIniState{\Aut}{\del{q}{\tu a}}$ accepts $\addPos{f}{\pair{\tu w}{T}}{b-1}$.
  Since $q$ is accessible from $\iniState$, there exists
  $\pair{\tu v}{S}\in\left(\digitDotSetDim\right)^{*}$ such that
  $\del{\iniState}{\pair{\tu v}{S}}=q$. Since $\changeIniState{\Aut}{\del{q}{\tu a}}$
  accepts $\addPos{f}{\pair{\tu w}{T}}{b-1}$ and
  $\del{\iniState}{\pair{\tu v}{S}}=q$, it follows that $\Aut$ accepts
  $\pair{\tu v}{S}\tu a\addPos{f}{\pair{\tu w}{T}}{b-1}$. Let us prove
  that
  $\wordToReal{\pair{\tu v}{S}\tu a\addPos{f}{\pair{\tu w}{T}}{b-1}}=
  \wordToReal{\pair{\tu v}{S}\tu a'\addPos{f}{\pair{\tu
        w}{T}}{0}}$. That is, for $j\in[d-1]$, we want to prove that
  $\wordToReal{\left(\pair{\tu v}{S}\tu a\addPos{f}{\pair{\tu
          w}{T}}{b-1}\right)_{i}}= \wordToReal{\left(\pair{\tu
        v}{S}\tu a'\addPos{f}{\pair{\tu w}{T}}{0}\right)_{i}}$.  For
  $i\in[d-1]\setminus f$, it suffices to remark that:
  \begin{equation*}
    \left(\pair{\tu v}{S}\tu a\addPos{f}{\pair{\tu w}{T}}{b-1}\right)_{i}=\pair{v_i}{S}a_{i}\pair{w_{i}}{T}=\pair{v_i}{S}a'_{i}\pair{w_{i}}{T}=\left(\pair{\tu v}{S}\tu a'\addPos{f}{\pair{\tu
          w}{T}}{0}\right)_{i}.
  \end{equation*}It remains to consider the case $i=f$. Remark that
  \begin{equation*}
    \left(\pair{\tu v}{S}\tu a\addPos{f}{\pair{\tu w}{T}}{b-1}\right)_{f}=\pair{v_f}{S}a_{f}\pair{(b-1)^{\omega}}{T}
  \end{equation*}
  and
  \begin{equation*}
    \left(\pair{\tu v}{S}\tu a'\addPos{f}{\pair{\tu
          w}{T}}{0}\right)_{f}=\pair{v_f}{S}a'_{f}\pair{0^{\omega}}{T}=\pair{v_f}{S}(a_{f}+1)\pair{0^{\omega}}{T}.
  \end{equation*}
  By \autoref{theo:rat}, both of those words encode the same number.

  Since $\Aut$ accepts
  $\pair{\tu v}{S}\tu a\addPos{f}{\pair{\tu w}{T}}{b-1}$,
  $\wordToReal{\pair{\tu v}{S}\tu a\addPos{f}{\pair{\tu w}{T}}{b-1}}=
  \wordToReal{\pair{\tu v}{S}\tu a'\addPos{f}{\pair{\tu w}{T}}{0}}$
  and $\Aut$ is saturated, $\Aut$ accepts
  $\pair{\tu v}{S}\tu a'\addPos{f}{\pair{\tu w}{T}}{0}$. It follows
  that $\Aut_{\del{q}{\tu a'}}$ accepts
  $\addPos{f}{\pair{\tu w}{T}}{0}$. Finally, since
  $\Aut_{\del{q}{\tu a'}}$ accepts $\addPos{f}{\pair{\tu w}{T}}{0}$,
  by Lemma \ref{lem:add-fix-par},
  $\left(\AutFix[\Aut]{f}{0}\right)_{\del{q}{\tu{a'}}}$ accepts
  $\pair{\tu w}{T}$.
  \paragraph{Property \eqref{lem:char-RVA-char} implies Property
    \eqref{lem:char-RVA-dist}}
  Let $\pair{\tu{w}}{\set{s}}$ and $\pair{\tu{w}'}{\set{s}}$ be two
  words as in Property \eqref{lem:char-RVA-dist}. Let us prove that
  $\Aut$ accepts $\tu{w'}$.  Let $f\in[d-1]$ be the only integer such
  that $w_{f}\ne w'_{f}$.  As explained in \autoref{theo:rat}, since
  $\wordToFractional{\pair{w_{f}}{\set s}}=
  \wordToFractional{\pair{w'_{f}}{\set s}}$,
  since ${w}_{f}\ne {w'}_{f}$, and since their natural parts have the
  same length, then
  $\set{w_{f},w'_{f}}=\set{u_{f}a_{f}(b-1)^{\omega},u_{f}(a_{f}+1)0^\omega}$,
  for some $u_{f}\in\digitSet^{*}$ and
  $a_{f}\in\digitSet\setminus\set{b-1}$. Let us assume that $w_{f}$ is
  $u_{f}a_{f}(b-1)^{\omega}$, the case where $w_{f}$ is
  $u_{f}(a_{f}+1)0^\omega$ is similar. Note that
  $w'_{f}=u_{f}(a_{f}+1)0^\omega$.

  Let $l$ be the length of the prefix of $(w_{f},\set s)$ before the
  occurrence of the letter $a_{f}$. It is $\length{u_{f}}$ if
  $s>\length{u_{f}}$, and $\length{u_{f}}+1$ otherwise.  Let
  $\pair{\tu u}{U}=\prefix{\pair{\tu w}{\set s}}{l}=\prefix{\pair{\tu
      w'}{\set s}}{l}$ be this prefix.  Let
  $\tu a= \pair{\tu w}{\set s}[l]$ and $\tu a'= (\tu w',\set
  s)[l]$. Similarly, let
  $\pair{\tu v}{V}=\suffix{\pair{\tu w}{\set s}}{l+1}$ and
  $\pair{\tu v'}{V}=\suffix{(\tu w',\set s)}{l+1}$ be the suffixes of
  $\tu w$ and $\tu w'$ after those occurrences of $\tu a$ and of
  $\tu a'$ respectively. 
  Note that the notations $u_{f}$, $a_{f}$ and $a'_{f}$
  introduced above are consistent with the notations $\tu u$, $\tu a$
  and $\tu a'$.  Note also that ${v}_{f}=(b-1)^{\omega}$, that
  ${v'}_{f}=0^{\omega}$, and that $v_{i}=v'_{i}$ for all
  $i\in[d-1]\setminus \set f$. Hence
  $\removePos{f}{\pair{\tu v}{V}}=\removePos{f}{\pair{\tu v'}{V}}$.
  The notation $a_{f}$ introduced above is coherent with the notation
  $\tu a$. Since $a'_{f}=a_{f}+1$, $\tu a$ and $\tu a'$ satisfy the
  hypothesis of Property \eqref{lem:char-RVA-char}. It follows that
  $\toInfWord{\changeIniState{\left(\AutFix{f}{(b-1)}\right)}{\del{q}{\tu{a}}}}=
  \toInfWord{\changeIniState{\left(\AutFix{f}{0}\right)}{\del{q}{\tu
        a'}}}$.

  We can now prove that $\pair{\tu w'}{\set s}$ is accepted by
  $\Aut$. Since
  $\pair{\tu w}{\set s}=\pair{\tu u}{U}\tu a\pair{\tu v}{V}$ is
  accepted by $\Aut$ and $v_{f}=(b-1)^{\omega}$, by
  \autoref{lem:ult-fix}, 
  $\changeIniState{\left(\AutFix{f}{(b-1)}\right)}{\del{\iniState}{\pair{\tu
        u}{U}\tu{a}}}$ accepts $\removePos{f}{\pair{\tu v}{V}}$.
  Since
  $\changeIniState{\left(\AutFix{f}{(b-1)}\right)}{\del{\iniState}{\pair{\tu
        u}{U}\tu{a}}}$ accepts $\removePos{f}{\pair{\tu v}{V}}$ and
  $\toInfWord{\changeIniState{\left(\AutFix{f}{(b-1)}\right)}{\del{q}{\tu{a}}}}=
  \toInfWord{\changeIniState{\left(\AutFix{f}{0}\right)}{\del{q}{\tu
        a'}}}$,
  $\changeIniState{\left(\AutFix{f}{0}\right)}{\del{\iniState}{\pair{\tu
        u}{U}\tu a'}}$ accepts $\removePos{f}{\pair{\tu v}{V}}$. Since
  $\removePos{f}{\pair{\tu v}{V}}=\removePos{f}{\pair{\tu v'}{V}}$ and
  $\changeIniState{\left(\AutFix{f}{0}\right)}{\del{\iniState}{\pair{\tu
        u}{U}\tu a'}}$ accepts $\removePos{f}{\pair{\tu v}{V}}$,
  $\changeIniState{\left(\AutFix{f}{0}\right)}{\del{\iniState}{\pair{\tu
        u}{U}\tu a'}}$ accepts $\removePos{f}{\pair{\tu v'}{V}}$. It
  follows from \autoref{lem:ult-fix} that
  $\Aut$ accepts
  $\pair{\tu u}{U}\tu a'\pair{\tu v}{V}=(\tu w',\set s)$.
  \paragraph{Property \eqref{lem:char-RVA-dist} implies Property
    \eqref{lem:char-RVA-nat-fixed}} We must prove that, for all
  $\pair{\tu w}{\set s}, \pair{\tu
    w'}{\set{s}}\in\left(\digitDotSetDim\right)^{\omega}$, if $\Aut$
  accepts $\pair{\tu w}{\set s}$ and
  $\wordToReal{\pair{\tu w}{\set s}}=\wordToReal{(\tu w',\set{s})}$
  then $\Aut$ accepts $\pair{\tu w'}{\set {s'}}$.  The proof is by
  induction on $i=\distWord{w}{w'}$.  The case $i=0$ is trivial, since
  it means that $\pair{\tu {w}}{\set s}=\pair{\tu{w'}}{\set {s}}$.
  Let us now assume that $i>1$ and that the induction hypothesis holds
  when $\distWord{w}{w'}<i$.

  Since $\distWord{w}{w'}=i$ and $i>1$, there exists $f\in[d-1]$ such
  that $w_{f}\ne w'_{f}$.  Let
  $\pair{\tu w''}{\set s}\in\left(\digitDotSetDim\right)^{\omega}$
  such that $w''_{f}=w_{f}$ and such that $w_{k}''=w'_{k}$ for all
  $k\in[b-1]\setminus\set f$. Note that $\distWord{w}{w''}=i-1$,
  $\distWord{w'}{w''}=1$ and
  $\wordToReal{\pair{\tu w}{\set s}}=\wordToReal{\pair{\tu {w''}}{\set
      s}}=\wordToReal{\pair{\tu {w'}}{\set s}}$.

  Since $\distWord{w}{w''}=i-1$ and
  $\wordToReal{\pair{\tu w}{\set s}}=\wordToReal{\pair{\tu {w''}}{\set
      s}}$, by induction hypothesis $\Aut$ accepts
  $\pair{\tu {w''}}{\set s}$.  Since $\distWord{w''}{w'}=1$ and
  $\wordToReal{\pair{\tu{w''}}{\set
      s}}=\wordToReal{\pair{\tu{w'}}{\set s}}$ and since $\Aut$
  accepts $\pair{\tu {w''}}{\set s}$, by Property
  \eqref{lem:char-RVA-dist}, $\Aut$ accepts $\pair{\tu{w'}}{\set s}$.
  \paragraph{Property
    \eqref{lem:char-RVA-nat-fixed} implies Property \eqref{lem:char-RVA-RVA}}
  We must prove that, for all
  $\pair{\tu w}{\set s}, (\tu
  w',\set{s'})\in\left(\digitDotSetDim\right)^{\omega}$, if $\Aut$
  accepts $\pair{\tu w}{\set s}$ and
  $\wordToReal{\pair{\tu w}{\set s}}=\wordToReal{\pair{\tu w'}{\set
      {s'}}}$ then $\Aut$ accepts $\pair{\tu w'}{\set {s'}}$.  Let us
  assume that $s\le s'$, the case $s>s'$ is similar. Since
  $\del{\iniState}{\tu 0}=\iniState$ and $\Aut$ accepts
  $(\tu w, \set s)$, $\Aut$ accepts $\pair{\tu
    0^{s'-s}}{\set{s'}}$. Note that
  $\wordToReal{\pair{\tu 0^{s'-s}\tu w}{\set{s'}}}=
  \wordToReal{\pair{\tu w}{\set s}}=\wordToReal{\pair{\tu w'}{\set
      {s'}}}$. Since
  $\wordToReal{\pair{\tu 0^{s'-s}\tu
      w}{\set{s'}}}=\wordToReal{\pair{\tu w'}{\set s'}}$ and $\Aut$
  accepts $\pair{\tu 0^{s'-s}\tu w}{\set{s'}}$, by Property
  \eqref{lem:char-RVA-nat-fixed}, $\Aut$ accepts
  $\pair{\tu w'}{\set {s'}}$.
\end{proof}

\section{Characterization of $d$-sequential automata}\label{sec:char-RVA-seq}
A characterization of $d$-sequential automata is now given. This
characterization is similar to Property \ref{lem:char-RVA-char} of
\autoref{prop:char-RVA}. Instead of doing the whole proof again for
sequential automata, we prove that this characterization is correct by
proving that it implies the characterization of
\autoref{prop:char-RVA} on the parallelization of the $d$-sequential
automata.
\begin{proposition}\label{theo:seq-char}
  Let $\Aut$ be a minimal weak $d$-sequential Büchi automaton.  The
  following statements are equivalent:
  \begin{enumerate}
  \item\label{theo:seq-char-RVA-RVA} The automaton $\Aut$ is a RVA.
  \item\label{theo:seq-char-RVA-char}
    \begin{itemize}
    \item $\del{\iniState}{0^{d}}=\iniState$ and
    \item For each $q\in Q$ accessible in $\Aut$ from $q_{0}$, for
      each $a\in\digitSet\setminus{b-1}$, $
        {\changeIniState{\left(\AutFixSeq{(b-1)}\right)}{(\del{q}{a},0)}}$ and
        ${\changeIniState{\left(\AutFixSeq{0}\right)}{(\del{q}{a+1},0)}}$
        accept the same language.
    \end{itemize}
  \end{enumerate}
\end{proposition}
In order to prove this proposition, we introduce the following
lemma. This lemma allows us to reduce Property
\eqref{theo:seq-char-RVA-char} of \autoref{theo:seq-char} to Property
\eqref{lem:char-RVA-char} of \autoref{prop:char-RVA}.
\begin{lemma}\label{lem:char-seq-to-paral}
  Let $\Aut$ be a weak $d$-sequential Büchi automaton, $f\in[d-1]$,
  $z\in\digitSet$, $q$ a state of $\Aut$, $\tu a\in\digitSetDim$ and
  $\pair{\tu w}{S}\in\left(\squarphabetDotDim\right)^{\omega}$. The
  word $\pair{\tu w}{S}$ is accepted by
  $\changeIniState{\left(\AutFix[\paral{\Aut}]{f}{z}\right)}{(\paralDel{q}{\tu{a}},0)}$
  if and only if
  $\changeIniState{\left(\AutFixSeq{z}\right)}{\del{q}{a_{0}\dots{}a_{f}}}$
  accepts $a_{f+1}\dots{}a_{d-1}\seqen{\pair{\tu w}{S}}$.
\end{lemma}
\begin{proof}
  By \autoref{lem:pemut-fix-state-seq}
  \quo{$\changeIniState{\left(\AutFix[\paral{\Aut}]{f}{z}\right)}{(\paralDel{q}{\tu{a}},0)}$
    accepts $\pair{\tu w}{S}$} is equivalent to:
  \quo{$\AutFix[\left(\changeIniState{\paral{\Aut}}{\paralDel{q}{\tu{a}}}\right)]{f}{z}$
    accepts $\pair{\tu w}{S}$}. By \autoref{lem:add-fix-par}, it is
  equivalent to:
  \quo{$\left(\changeIniState{\paral{\Aut}}{\paralDel{q}{\tu{a}}}\right)$
    accepts $\addPos{f}{\pair{\tu w}{S}}{z}$}, hence to
  \quo{$\left(\paral{\Aut}\right)_{q}$ accepts
    $\tu a\addPos{f}{\pair{\tu w}{S}}{z}$}.  By
  \autoref{lem:paral-state}, it is also equivalent to:
  \quo{$\paral{\changeIniState{\Aut}{q}}$ accepts
    $\tu a\addPos{f}{\pair{\tu w}{S}}{z}$}.  By \autoref{lem:seq-aut},
  it is also equivalent to: \quo{$\changeIniState{\Aut}{q}$ accepts
    $\seqen{\tu a\addPos{f}{\pair{\tu w}{S}}{z}}$}. Note that
  $\seqen{\tu a\addPos{f}{\pair{\tu w}{S}}{z}}$ equals
  $a_{0}\dots{}a_{d-1}\seqen{\addPos{f}{\pair{\tu w}{S}}{z}}$, thus
  the above-mentionned properties are equivalent to:
  \quo{$\changeIniState{\Aut}{q}$ accepts
    $a_{0}\dots{}a_{d-1}\seqen{\addPos{f}{\pair{\tu w}{S}}{z}}$} and
  then to: \quo{$\Aut_{\del{q}{a_{0}\dots,a_{f}}}$ accepts
    $a_{f+1}\dots{}a_{d-1}\seqen{\addPos{f}{\pair{\tu w}{S}}{z}}$}.
  By \autoref{lem:remove-fix-seq}, it is equivalent to:
  \quo{$\AutFixSeq[\left(\Aut_{_{\del{q}{a_{0}\dots,a_{f}}}}\right)]{z}$
    accepts
    $\removePos{d-1}{a_{f+1}\dots{}a_{d-1}\seqen{\addPos{f}{\pair{\tu
            w}{S}}{z}}}$}.

  Note that the length of $a_{f+1}\dots{}a_{d-1}$ is
  $(d-1)-(f+1)+1=d-f-1$, and that $(d-1)-(d-f-1)=f$, thus by
  \autoref{lem:change-move},
  $\removePos{d-1}{a_{f+1}\dots{}a_{d-1}\seqen{\addPos{f}{\pair{\tu
          w}{S}}{z}}}$
  equals
  $a_{f+1}\dots{}a_{d-1}\removePos{f}{\seqen{\addPos{f}{\pair{\tu
          w}{S}}{z}}}$.
  By definition of $\fix_{\square\MVAt f}$ on $d$-sequential number,
  $\removePos{f}{\seqen{\addPos{f}{\pair{\tu w}{S}}{z}}}$ equals
  $\seqen{\removePos{f}{\addPos{f}{\pair{\tu w}{S}}{z}}}$.  By
  \autoref{lem:change-twice},
  $\removePos{f}{\addPos{f}{\pair{\tu w}{S}}{z}}$ equals
  $\removePos{f}{{\pair{\tu w}{S}}}$.  Since $w_{f}=\square^{\omega}$,
  by \autoref{lem:change-cancel},
  $\removePos{f}{{\pair{\tu w}{S}}}=\pair{\tu w}{S}$.  It follows that
  $\removePos{d-1}{a_{f+1}\dots{}a_{d-1}\seqen{\addPos{f}{\pair{\tu
          w}{S}}{z}}}$
  equals $a_{f+1}\dots{}a_{d-1}\seqen{\pair{\tu w}{S}}$.  Thus, all
  the above mentionned facts are equivalent to:
  \quo{$\AutFixSeq[\left(\Aut_{_{\del{q}{a_{0}\dots,a_{f}}}}\right)]{b-1}$
    accepts $a_{f+1}\dots{}a_{d-1}\seqen{\pair{\tu w}{S}}$}. Finally,
  by \autoref{lem:pemut-fix-state-par}, it is equivalent to
  \quo{$\changeIniState{\left(\AutFixSeq{z}\right)}{\del{q}{a_{0}\dots{}a_{f}}}$
    accepts $a_{f+1}\dots{}a_{d-1}\seqen{\pair{\tu w}{S}}$}.
\end{proof}
\autoref{theo:seq-char} is now proven.
\begin{proof}
  Let us show that Property \eqref{theo:seq-char-RVA-RVA} implies
  Property \eqref{theo:seq-char-RVA-char}.  The proof of the first
  part of this Property is the same than the proof of
  \autoref{lem:init-0}.  The proof of the second part of this Property
  is the same than the proof that Property \eqref{lem:char-RVA-RVA} of
  \autoref{prop:char-RVA} implies Property \eqref{lem:char-RVA-char}
  of \autoref{prop:char-RVA}.

  It remains to prove that Property \eqref{theo:seq-char-RVA-char}
  implies Property \eqref{theo:seq-char-RVA-RVA}. We want to prove
  that $\Aut$ is a RVA. Since $\Aut$ is $d$-sequential, it remains to
  prove that $\toInfWord{\Aut}$ is saturated. It suffices to prove
  that $\paral{\toInfWord{\Aut}}$ is saturated. By
  \autoref{lem:paral-aut}, $\paral{\Aut}$ accepts
  $\paral{\toInfWord{\Aut}}$. Therefore, we only have to prove that
  $\paral{\Aut}$ is a RVA. By \autoref{prop:char-RVA}, it suffices to
  prove that $\paralDel{\iniState}{\tu 0}=\iniState$ and that, for
  each $f\in[d-1]$, for each $q\in Q$, accessible in $\Aut$ from
  $q_{0}$, for each $\tu a\in\digitSetDim$ with $a_{f}<b-1$,
  $
  \toInfWord{\changeIniState{\left(\AutFix[{\paral{\Aut}}]{f}{(b-1)}\right)}{\paralDel{q}{\tu{a}}}}=
  \toInfWord{\changeIniState{\left(\AutFix[\paral{\Aut}]{f}{0}\right)}{\paralDel{q}{\tu
        a'}}}$ where
  $\tu a'=a_{0}\dots{}a_{f-1}(a_{f}+1)a_{f+1}\dots{}a_{d-1}$.

  Note that the first part of Property \eqref{theo:seq-char-RVA-char}
  clearly implies $\paralDel{\iniState}{\tu 0}=\iniState$.  Let $q$,
  $\tu a$, $\tu a'$ as above, it remains to prove that
  $\toInfWord{\changeIniState{\left(\AutFix[\paral{\Aut}]{f}{(b-1)}\right)}{\paralDel{q}{\tu{a}}}}\subseteq
  \toInfWord{\changeIniState{\left(\AutFix[\paral{\Aut}]{f}{0}\right)}{\paralDel{q}{\tu{a'}}}}$,
  the reverse inclusion is similar.  Let
  $\pair{\tu w}{S}\in\squarphabetDotDim^{\omega}$ accepted by
  $\changeIniState{\left(\AutFix[\paral{\Aut}]{f}{(b-1)}\right)}{\paralDel{q}{\tu{a}}}$,
  we want to prove that it is accepted by
  $\AutFix[\left(\paral{\Aut}\right)_{\paralDel{q}{\tu{a}'}}]{f}{0}$.

  Since
  $\changeIniState{\left(\AutFix[\paral{\Aut}]{f}{(b-1)}\right)}{\paralDel{q}{\tu{a}}}$
  accepts $\pair{\tu w}{S}\in\left(\squarphabetDotDim\right)^{\omega}$, by
  \autoref{lem:char-seq-to-paral},
  $\AutFixSeq[\left(\Aut_{\del{q}{a_{0}\dots{}a_{f-1}a_{f}}}\right)]{b-1}$
  accepts $a_{f+1}\dots{}a_{d-1}\seqen{\pair{\tu w}{S}}$. By the
  second part of Property \eqref{theo:seq-char-RVA-char}, it follows
  that
  $\AutFixSeq[\left(\Aut_{\del{q}{a_{0}\dots{a_{f-1}}(a_{f}+1)}}\right)]{0}$
  accepts $a_{f+1}\dots{}a_{d-1}\seqen{\pair{\tu w}{S}}$. By
  \autoref{lem:char-seq-to-paral}, it follows that
  $\changeIniState{\left(\AutFix[\paral{\Aut}]{f}{0}\right)}{\paralDel{q}{\tu{a'}}}$
  accepts $\pair{\tu w}{S}$.
\end{proof}
\section{Algorithms}\label{sec:algo-RVA}
We now show how to decide efficiently whether an automaton is a RVA.
\begin{theorem}\label{theo:RVA-par}
  Let $\Aut=\autPar{Q}{\digitDotSetDim}{\delta}{q_0}{F}$ an automaton
  with $n$ states. It is decidable in time $\bigO{n\log(n)db^{d}}$ and
  space $\bigO{nb^{d}}$ whether $\Aut$ is a RVA.
\end{theorem}
Note that $b^{d}$ is the cardinality of the alphabet. Thus this
algorithm is quasi-linear in the size of its input.
\begin{proof}
  Without loss of generality, it can be assumed that the automaton is
  minimal.  The algorithm consists in three parts. First, the
  algorithm checks whether the algorithm of \autoref{theo:d-par}
  applied on $\Aut$ returns true. Secondly, the algorithm checks
  whether $\del{\iniState}{\tu 0}=\iniState$. Thirdly, the algorithm
  runs on each $f\in[d-1]$. For each $f$, the algorithm generates the
  automata $\AutFix{f}{0}$ and $\AutFix{f}{(b-1)}$ and the data
  structure mentionned in \autoref{theo:min-quasi}. The algorithm runs
  on each $q\in Q$ and $\tu a$, $\tu a'$ as in Property
  \eqref{lem:char-RVA-char} of \autoref{prop:char-RVA}. The algorithm
  then applies the algorithm of \autoref{theo:min-quasi} to check
  whether, for all pairs $(\paralDel{q}{\tu a},\paralDel{q}{\tu a'})$,
  $\changeIniState{\left(\AutFix{f}{b-1}\right)}{\paralDel{q}{\tu a}}$
  and
  $\changeIniState{\left(\AutFix{f}{0}\right)}{\paralDel{q}{\tu a'}}$
  accept the same language.  If one of those checks fail the algorithm
  rejects. Otherwise, the algorithm accepts.

  By \autoref{lem:fix-weak}, $\AutFix{f}{0}$ and $\AutFix{f}{(b-1)}$
  are weak, thus, \autoref{theo:min-quasi} can be used.  It follows
  from \autoref{theo:d-par}, \autoref{lem:init-0} and
  \autoref{prop:char-RVA} that this algorithm accepts exactly the
  $d$-parallel automata which are RVA.

  It remains to consider the complexity. By \autoref{theo:d-par}, the
  first part runs in $\bigO{nb^{d}}$ and space $\bigO{n}$. The second
  part clearly runs in constant time and space. By
  \autoref{theo:min-quasi}, the last part runs in time
  $\bigO{n\log(n)b^{d}d}$ and space $\bigO{n\log{n}b^{d}}$.
\end{proof}
We now prove a similar theorem for $d$-sequential automata.
\begin{theorem}\label{theo:RVA-seq}
  Let $\Aut=\autPar{Q}{\digitDotSet}{\delta}{q_0}{F}$ an automaton
  with $n$ states. It is decidable in time $\bigO{nd\log(nd)b}$ and
  space $\bigO{ndb}$ whether $\Aut$ is a RVA.
\end{theorem}
\begin{proof}
  Without loss of generality, it can be assumed that the automaton is
  minimal.  The algorithm and its proof are similar to the ones of
  \autoref{theo:RVA-par}. The algorithm consists in three
  parts. Firstly, the algorithm applies the algorithm of
  \autoref{theo:d-seq} to check whether $\Aut$ is a $d$-sequential
  automaton. Secondly, it checks whether
  $\del{\iniState}{0^{d}}=\iniState$. Thirdly, the algorithm generates
  the automata $\AutFixSeq{0}$ and $\AutFixSeq{(b-1)}$ and the data
  structure mentionned in \autoref{theo:min-quasi}.  The algorithm
  runs on each $q\in Q$ and $a\in\digitSet\setminus{b-1}$. For each
  $q$ and $a$, the algorithm applies the algorithm of
  \autoref{theo:min-quasi} to check whether
  ${\changeIniState{\left(\AutFixSeq{(b-1)}\right)}{\del{q}{a}}}$ and
  ${\changeIniState{\left(\AutFixSeq{0}\right)}{\del{q}{a+1}}}$ accept
  the same language.

  By \autoref{lem:fix-weak}, $\AutFixSeq{0}$ and $\AutFixSeq{(b-1)}$
  are weak, thus, \autoref{theo:min-quasi} can be used. It follows
  from \autoref{theo:d-seq} and \autoref{theo:seq-char} that this
  algorithm accepts exactly the $d$-sequential automata which are RVAs.
  
  It remains to consider the complexity.  By \autoref{theo:d-seq}, the
  first part runs in time $\bigO{ndb}$ and space $\bigO{nd}$. The
  second part runs in time $\bigO{d}$ and constant space. Finally, by
  \autoref{theo:min-quasi}, the third part runs in time
  $\bigO{nd\log(nd)b}$ and space $\bigO{ndb}$.



\end{proof}
A last algorithm is given for the special case of dimension $d=1$.
\begin{theorem}\label{theo:dim-1}
  Let $\Aut=\autPar{Q}{\digitSet}{\delta}{q_0}{F}$ a minimal weak
  Büchi automaton with $n$ states. It is decidable in time $\bigO{nb}$
  and space $\bigO{n}$ whether $\toInfWord{\Aut}$ is saturated.
\end{theorem}
Note that $\Aut$ is assumed minimal.
\begin{proof}
  Note that $\squarphabetDim[b][1]$ contains exactly one letter,
  $\left(\square\right)$, thus both languages of Equation
  \eqref{eq:char-RVA} are either $\emptyset$ or
  $\left(\square\right)^{\omega}$. It follows that Equation
  \eqref{eq:char-RVA} holds if and only if
  $\del{q}{a}\in\emptyStates^{(b-1)}$ is equivalent to
  $\del{q}{a+1}\in\emptyStates^{0}$, where $\emptyStates^{z}$ is the
  set of states $q$ such that $\changeIniState{\AutFix 0z}q$ accepts
  the empty language.
    
  The algorithm is now given. The algorithm checks whether
  $\del{\iniState}{0}=\iniState$. The algorithm computes
  $\AutFix{0}{(b-1)}$ and $\AutFix{0}{0}$. The algorithm applies the
  algorithm of \autoref{lem:alg:sets} to those two automata in order
  to compute the sets $\emptyStates^{(b-1)}$ and
  $\emptyStates^{0}$. The algorithm runs on each state $q\in Q$
  accessible from $\iniState$, and on each
  $a\in\digitSet\setminus \set{b-1}$. For each $q$ and $a$, the
  algorithm checks whether $\del{q}{a}\in\emptyStates^{(b-1)}$ is
  equivalent to $\del{q}{a+1}\in\emptyStates^{0}$.  Finally, if one of
  the checks fail, the algorithm rejects, otherwise it accepts.

  Let us consider the complexity. The automata $\AutFix{0}{(b-1)}$ and
  $\AutFix{0}{0}$ clearly takes space $\bigO{n}$. Applying the
  algorithm of \autoref{lem:alg:sets} takes time $\bigO{n}$ and takes
  space $\bigO{n}$. For a set $q$ and a letter $a$, checking the
  equivalence is done in constant time and space. Thus, the final loop
  runs in time $\bigO{nb}$ and constant space. Finaly, the whole
  algorithm runs in time $\bigO{nb}$ and space $\bigO{n}$.
\end{proof}
\section{Considering negative reals}\label{sec:neg}
In this section, we consider the case of negative numbers. Given
$\natPart{w}\in\digitSet^{*}$, $aw$ encodes, in $b$-complement
representation, the number $\wordToNatural{\natPart{w}}$ if $a=0$,
$-b^{\length{\natPart{w}}}+\wordToNatural{\natPart{w}}$ if $a=(b-1)$,
and is undefined otherwise. Similarly, given
$\fraPart{w}\in\digitSet^{\omega}$,
$\wordToReal{a\natPart{w}\realDot\fraPart{w}}$ encodes, in 
$b$-complement representation, the number encoded by $a\natPart{w}$,
plus $\wordToFractional{\fraPart{w}}$.

When considering $b$-complement representation, \autoref{theo:rat}
must be changed as follows. For $l$ great enough, a real $q$ has two
encoding whose natural part's length is $l$ if and only if it is of
the form $nb^{p}$ with $n,p\in\Z$. If $q=0$, those encodings are
$(b-1)^{l}\realDot(b-1)^{\omega}$ and $0^{l}\realDot0^{\omega}$,
otherwise they are as in Equation \eqref{eq:theo:rat}.

A characterization of automata accepting saturated languages in
$b$-complement representation is now given. This characterization and
its proof is similar to the ones of \eqref{prop:char-RVA}.
\begin{proposition}
  Let $\Aut=\autPar{Q}{\digitDotSetDim}{\delta}{q_0}{F}$ a weak Büchi
  automaton over alphabet $\digitDotSetDim$. It accepts a saturated
  language in $b$-complement if and only if:
  \begin{enumerate}
  \item\label{relat-copy} $\del{\iniState}{\tu a\tu a}=\del{\iniState}{\tu a}$ for all
    $\tu a\in\set{0,(b-1)}^{d}$,
  \item $\del{\iniState}{\tu a}\in\emptyStates$ for all
    $\tu a\in\digitDotSetDim\setminus\left(\set{0,(b-1)}^{d}\right)$,
  \item for each $f\in[d-1]$, for each $q\in Q\setminus \iniState$,
    accessible in $\Aut$ from $q_{0}$, for each $\tu a\in\digitSetDim$
    with $a_{f}<b-1$, 
    \begin{equation*}
      \toInfWord{\changeIniState{\left(\AutFix{f}{(b-1)}\right)}{\del{q}{\tu{a}}}}=
\toInfWord{\changeIniState{\left(\AutFix{f}{0}\right)}{\del{q}{\tu
a'}}},
    \end{equation*} where $\tu
a'=a_{0}\dots{}a_{f-1}(a_{f}+1)a_{f+1}\dots{}a_{d-1}$
  \item\label{relat:0} \begin{equation*}
\toInfWord{\left(\AutFix{f}{(b-1)}\right)}=
\toInfWord{\left(\AutFix{f}{0}\right)},
    \end{equation*}
  \end{enumerate}
\end{proposition} Property \eqref{relat:0} allows to consider the case
of the real 0. Note that Property \eqref{relat-copy} considers $b^{d}$
letters. It implies that this proposition does not lead to a
polynomial time algorithm in the case of $d$-sequential automata. In
the case of $d$-parallel automata, this proposition leads easily to
algorithms, as \autoref{theo:seq-char} led to \autoref{theo:RVA-par}
and to \autoref{theo:dim-1}.
\begin{theorem} Let $\Aut=\autPar{Q}{\digitDotSetDim}{\delta}{q_0}{F}$
  an automaton with $n$ states reading reals in $b$-complement. It is
  decidable in time $\bigO{n\log(n)db^{d}}$ and space $\bigO{nb^{d}}$
  whether $\Aut$ is a RVA.
\end{theorem}
\begin{theorem} Let $\Aut=\autPar{Q}{\digitDotSet}{\delta}{q_0}{F}$ a
  minimal weak Büchi automaton with $n$ states reading reals in
  $b$-complement. It is decidable in time $\bigO{nb}$ and space
  $\bigO{n}$ whether $\toInfWord{\Aut}$ accepts a saturated language
  in $b$-complement.
\end{theorem}
\section{Conclusion} In this paper, we have proven that it is
decidable in quasi-linear time whether a weak Büchi automaton reading
digits and dots accept a language which encode a saturated set of
vector reals.

Two natural questions remain open.


Can this algorithm be adapted for some classes of automata which are
not weak. Even in the case of dimension 1, it seems complicated to
test whether
$\toInfWord{\changeIniState{\Aut}{\del{q}{0}}}=\toInfWord{\Aut}$, when
the automaton is not weak.

Given an automaton $\Aut$ which accept a set
$R\subseteq \left(\mathbb R^{\ge0}\right)^{d}$, is there some
efficient way to compute a saturated automaton $\Aut'$ which also
accept $R$. One could compute a $\fo{\R,\Z;X_{b},+,<}$-formula
defining $R$, and from this formula a saturated Büchi
automaton. However, this method is inneficient, and does not preserve
weakness.

\bibliographystyle{alpha} \bibliography{../fo}
  
\renewenvironment{theindex}{\begin{multicols}{2}\begin{itemize}}{\end{itemize}\end{multicols}}
\printindex

\end{document}

%% file: parent-preambule.tex
\usepackage[utf8]{inputenc}
\usepackage[T1]{fontenc}
\usepackage[nottoc,notlof, notlot,section]{tocbibind} 

\usepackage{xparse}

\usepackage{appendix}
\usepackage{afterpage}
\usepackage{multicol}
\usepackage{color}
\usepackage[usenames]{xcolor}
\usepackage{xspace}
\usepackage{setspace}
\usepackage{multido}
\usepackage[normalem]{ulem}
\usepackage{soul}
\usepackage{url}

\usepackage{color}
\usepackage{csquotes}
\usepackage[inline]{enumitem}
\usepackage{wrapfig}
\usepackage{framed}
\usepackage[hypcap]{caption}
\usepackage{subcaption}
\usepackage{tikz}
\usetikzlibrary{positioning,shadows,arrows,automata}
\usepackage{graphics}
\usepackage{graphicx}

\usepackage{multirow}
\usepackage{makecell}
\usepackage{etoolbox}
\usepackage{times}
\usepackage{mathtools}
\usepackage{amsmath}
\usepackage{mathtools}
\usepackage{amssymb} 
\usepackage{amsthm}
\usepackage{centernot}

\usepackage[cyr]{aeguill}
\usepackage{tipa}
\usepackage{stmaryrd} 

\usepackage[boxruled,linesnumbered,noend]{algorithm2e}
\usepackage{algorithmic}
\usepackage{complexity}
\newclass{\MSO}{MSO}
\newclass{\WMSO}{WMSO}
\newclass{\SO}{SO}
\usepackage{latexsym} 
\usepackage{gensymb}

\newcommand{\ignore}[1]{}
\let\originalleft\left
\let\originalright\right
\renewcommand{\left}{\mathopen{}\mathclose\bgroup\originalleft}
\renewcommand{\right}{\aftergroup\egroup\originalright}

\newcommand{\mb}{\mathbb}
\newcommand{\mc}{\mathcal}

\newcommand{\fo}[1]{\FO\left[#1\right]}

\newcommand{\bigO}[1]{O\left(#1\right)}

\DeclareDocumentCommand{\wordToReal}{ O{b}  m}{\wordToNumber[#1]{#2}^{\R}}
\DeclareDocumentCommand{\wordToFractional}{ O{b}  m}{\wordToNumber[#1]{#2}^{\fract}}
\DeclareDocumentCommand{\wordToNatural}{ O{b}  m}{\wordToNumber[#1]{#2}^{\nat}}
\DeclareDocumentCommand{\wordToNumber}{ O{b}  m}{
  \left[{#2}\right]_{#1}
}
\DeclareDocumentCommand{\wordToTupleReal}{ O{b} O{d} m}{\wordToNumber[#1][#2]{#3}^{\R}}
\DeclareDocumentCommand{\wordToTupleFractional}{ O{b} O{d} m}{\wordToNumber[#1][#2]{#3}^{\fract}}
\DeclareDocumentCommand{\wordToTupleNatural}{ O{b} O{d} m}{\wordToNumber[#1][#2]{#3}^{\nat}}
\DeclareDocumentCommand{\wordToTupleNumber}{ O{b} O{d} m}{
  \def\temp{#2}\ifx\temp\empty
  \left[#3\right]_{#1}
  \else
  \left[#3\right]_{#1,#2}
  \fi
}
\DeclareDocumentCommand{\logicToTupleNumber}{ O{d} m}{
  \left[#2\right]_{#1}
}
\DeclareDocumentCommand{\autToTupleNumber}{ O{b} O{d} m}{
  \def\temp{#2}\ifx\temp\empty
  \left[#3\right]_{#1}
  \else
  \left[#3\right]_{#1,#2}
  \fi
}

\newcommand{\prefix}[2]{#1\left[<#2\right]}
\newcommand{\suffix}[2]{#1\left[\ge{}#2\right]}
\newcommand{\realDot}{\star}
\DeclareMathOperator{\paralOp}{par}
\newcommand{\paral}[2][d]{\paralOp_{#1}\left(#2\right)}
\DeclareMathOperator{\seqenOp}{seq}
\newcommand{\seqen}[2][d]{\seqenOp_{#1}\left(#2\right)}

\newcommand{\set}[1]{\left\{ #1\right\}}
\newcommand{\tuple}[1]{\left(#1\right)}

\newcommand{\N}{\mathbb N{}}
\newcommand{\nat}{I}
\newcommand{\fract}{F}
\newcommand{\natPart}[1]{#1_{\nat}}
\newcommand{\fraPart}[1]{#1_{\fract}}
\newcommand{\alphabet}{A}
\DeclareDocumentCommand{\digitSet}{ O{b}}{\Sigma_{#1}}
\DeclareDocumentCommand{\digitSetDim}{ O{b} O{d}}{\Sigma_{#1,#2}}
\DeclareDocumentCommand{\digitDotSetDim}{ O{b} O{d}}{\Sigma_{#1,#2}\cup\set{\realDot}}
\newcommand{\digitDotSet}[1][b]{\Sigma_{#1}\cup\set\realDot}
\newcommand{\Z}{\mathbb Z{}}

\renewcommand{\R}{\ensuremath{\mb{R}}}

\newcommand{\tu}[1]{\boldsymbol{#1}}

\newcommand{\ceil}[1]{\left\lceil{#1}\right\rceil}
\newcommand{\length}[1]{\left|{#1}\right|}
\newcommand{\card}[1]{|#1|}

\newcommand{\del}[2]{\delta(#1,#2)}

\newcommand{\Aut}{\mc{A}}
\newcommand{\autPar}[5]{\left(#1,#2,#3,#4,#5\right)}
\newcommand{\aut}{\autPar Q\alphabet\delta{q_{0}}F}

\newcommand{\changeIniState}[2]{#1_{#2}}
\newcommand{\emptyStates}[1][\Aut]{Q_{\emptyset}}
\newcommand{\modStates}[2][\Aut]{Q_{#2}}

\newcommand{\fraStates}[1][\Aut]{Q_{\fract}}

\newcommand{\emptyState}[1][\Aut]{q_{\emptyset,#1}}
\newcommand{\inftyState}[1][\Aut]{q_{\infty,#1}}
\newcommand{\zuState}[1][\Aut]{q_{[0,1],#1}}
\newcommand{\iniState}{q_{0}}

\usepackage{xstring}
\newcommand{\appendExp}[1]{
  \hspace{.5cm}
  \left(#1\right)
}
\newcommand{\anot}[3]{
  &#3&#2&
  \def\temp{#1}\ifx\temp\empty
  \else\appendExp{#1}
  \fi
}

\usepackage{aliascnt}
\usepackage{hyperref}

%% file: parent-en-preambule.tex
\newtheorem{theorem}{Theorem}[section]

\newaliascnt{lemma}{theorem}  
\newtheorem{lemma}[lemma]{Lemma}  
\aliascntresetthe{lemma}

\newaliascnt{corollary}{theorem}  
\newtheorem{corollary}[corollary]{Corollary}  
\aliascntresetthe{corollary}

\theoremstyle{definition}

\newaliascnt{example}{theorem}  
\newtheorem{example}[example]{Example}  
\aliascntresetthe{example}

\newaliascnt{definition}{theorem}  
\newtheorem{definition}[definition]{Definition}  
\aliascntresetthe{definition}

\newaliascnt{notation}{theorem}  
\newtheorem{notation}[notation]{Notation}  
\aliascntresetthe{notation}

\newaliascnt{proposition}{theorem}  
\newtheorem{proposition}[proposition]{Proposition}  
\aliascntresetthe{proposition}

\newaliascnt{property}{theorem}  
  
\aliascntresetthe{property}

\theoremstyle{remark}
\newaliascnt{remark}{theorem}  
  
\aliascntresetthe{remark}

\newcommand{\quo}[1]{\text{``}#1\text{''}}



%% file: preambule.tex
\usepackage{ marvosym }
\usepackage{ mathrsfs }

\newcommand{\AR}{\ARPar{R}}
\newcommand{\ARPar}[1]{\Aut_{#1}}

\newcommand{\toInfWord}[1]{L_{\omega}\left({#1}\right)}

\newcommand{\AutFix}[3][\Aut]{#1^{#3\MVAt{} #2}}
\newcommand{\AutFixSeq}[2][\Aut]{#1^{#2}}
\newcommand{\FixTran}[3][\delta]{#1^{#3\MVAt{} #2}}
\newcommand{\FixTranSeq}[2][\delta]{#1^{#2}}
\newcommand{\removePos}[2]{\changePos{#1}{#2}\square}
\newcommand{\addPos}[3]{\changePos{#1}{#2}{#3}}
\DeclareMathOperator{\fix}{fix}
\newcommand{\changePos}[3]{\fix^{#3\MVAt{} #1}\left(#2\right) }

\newcommand{\distWord}[2]{\card{\set{j\mid #1_j\ne{}#2_j}}}

\newcommand{\paralDel}[2]{\paral{\delta}(#1,#2)}
\newcommand{\pair}[2]{\left\langle #1,#2\right\rangle}

\DeclareDocumentCommand{\squarphabetDim}{ O{b} O{d}
  O{f}}{\digitSetDim[#1][#2]^{\square\MVAt{}#3}}
\DeclareDocumentCommand{\squarphabetDotDim}{ O{b} O{d}
  O{f}}{\digitSetDim[#1][#2]^{\square\MVAt{}#3}\cup\set{\realDot}}
\DeclareDocumentCommand{\squarphabet}{ O{b} }{\digitSet[#1]\cup\set\square}
\DeclareDocumentCommand{\squarphabetDot}{ O{b}  }{\digitSet[#1]\cup\set{\realDot,\square}}
